\documentclass{llncs}
\usepackage[dvips]{epsfig}
\usepackage{amsmath,amssymb,amsfonts,enumerate,bbm}
\usepackage{caption}
\usepackage{url}
\usepackage{ifthen}

\usepackage{graphicx}

\usepackage{subfigure}

\usepackage{color,soul}

\newtheorem{observation}[theorem]{Observation}
\newtheorem{fact}[theorem]{Fact}
\newcommand{\R}{\mathbb R}

\newcommand{\SINR}[0]{\mathrm{SINR}}

\newcommand{\NoRegret}[0]{\mathcal N}

\newcommand{\Noise}[0]{\eta}

\newcommand{\E}{\mathop{\mathbb E}}

\def\inline#1:{\par\vskip 7pt\noindent{\bf #1:}\hskip 10pt}

\def\blackslug{\hbox{\hskip 1pt \vrule width 4pt height 8pt
    depth 1.5pt \hskip 1pt}}
\def\QED{\quad\blackslug\lower 8.5pt\null\par}

\begin{document}

\title{Braess's Paradox in Wireless Networks: The Danger of Improved Technology}

\author{
Michael Dinitz \inst{1,2}
\and
Merav Parter \inst{1}
\thanks{Recipient of the Google European Fellowship in distributed computing;
research is supported in part by this Fellowship.}
}

\institute{Department of Computer Science and Applied Mathematics, \\
The Weizmann Institute, Rehovot, Israel\\
\email{\{michael.dinitz, merav.parter\}@weizmann.ac.il}
\thanks{
Supported in part by the Israel Science Foundation (grant 894/09),
the I-CORE program of the Israel PBC and ISF (grant 4/11),
the United States-Israel Binational Science Foundation (grant 2008348),
the Israel Ministry of Science and Technology (infrastructures grant),
and the Citi Foundation.}
\and
Department of Computer Science, Johns Hopkins University}

\maketitle

\begin{abstract}
When comparing new wireless technologies, it is common to consider the effect that they have on the capacity of the network (defined as the maximum number of simultaneously satisfiable links).  For example, it has been shown that giving receivers the ability to do interference cancellation, or allowing transmitters to use power control, never decreases the capacity and can in certain cases increase it by $\Omega(\log (\Delta \cdot P_{\max}))$, where $\Delta$ is the ratio of the longest link length to the smallest transmitter-receiver distance and $P_{\max}$ is the maximum transmission power.  But there is no reason to expect the optimal capacity to be realized in practice, particularly since maximizing the capacity is known to be NP-hard.  In reality, we would expect links to behave as self-interested agents, and thus when introducing a new technology it makes more sense to compare the values reached at game-theoretic equilibria than the optimum values.

In this paper we initiate this line of work by comparing various notions of equilibria (particularly Nash equilibria and no-regret behavior) when using a supposedly ``better" technology.  We show a version of Braess's Paradox for all of them: in certain networks, upgrading technology can actually make the equilibria \emph{worse}, despite an increase in the capacity.  We construct instances where this decrease is a constant factor for power control, interference cancellation, and improvements in the SINR threshold ($\beta$), and is $\Omega(\log \Delta)$ when power control is combined with interference cancellation.  However, we show that these examples are basically tight: the decrease is at most $O(1)$ for power control, interference cancellation, and improved $\beta$, and is at most $O(\log \Delta)$ when power control is combined with interference cancellation.
\end{abstract}


\section{Introduction}

Due to the increasing use of wireless technology in communication networks, there has been a significant amount of research on methods of improving wireless performance.  While there are many ways of measuring wireless performance, a good first step (which has been extensively studied) is the notion of \emph{capacity}.  Given a collection of communication links, the capacity of a network is simply the maximum number of simultaneously satisfiable links.  This can obviously depend on the exact model of wireless communication that we are using, but is clearly an upper bound on the ``usefulness" of the network.  There has been a large amount of research on analyzing the capacity of wireless networks (see e.g.~\cite{GuKu00,GHWW09,AD09,K11}), and it has become a standard way of measuring the quality of a network.  Because of this, when introducing a new technology it is interesting to analyze its affect on the capacity.  For example, we know that in certain cases giving transmitters the ability to control their transmission power can increase the capacity by $\Omega(\log \Delta)$ or $\Omega(\log P_{\max}))$  \cite{ALP09}, where $\Delta$ is the ratio of the longest link length to the smallest transmitter-receiver distance, and can clearly never decrease the capacity. 

However, while the capacity might improve, it is not nearly as clear that the \emph{achieved} capacity will improve.  After all, we do not expect our network to actually have performance that achieves the maximum possible capacity.  We show that not only might these improved technologies not help, they might in fact \emph{decrease} the achieved network capacity.  Following Andrews and Dinitz~\cite{AD09} and \'Asgeirsson and Mitra~\cite{Asgeirsson11}, we model each link as a self-interested agent and analyze various types of game-theoretic behavior (Nash equilibria and no-regret behavior in particular).  We show that a version of \emph{Braess's Paradox}~\cite{Braess68} holds: adding new technology to the networks (such as the ability to control powers) can actually decrease the average capacity at equilibrium.


\subsection{Our Results}

Our main results show that in the context of wireless networks, and particularly in the context of the SINR model, there is a version of \emph{Braess's Paradox}~\cite{Braess68}.  In his seminal paper, Braess studied congestion in road networks and showed that adding additional roads to an existing network can actually make congestion \emph{worse}, since agents will behave selfishly and the additional options can result in worse equilibria.  This is completely analogous to our setting, since in road networks adding extra roads cannot hurt the network in terms of the value of the optimum solution, but can hurt the network since the \emph{achieved} congestion gets worse.  In this work we consider the physical model (also called the SINR model), pioneered by Moscibroda and Wattenhofer~\cite{MW06} and described more formally in Section~\ref{sec:models}.  Intuitively, this model works as follows: every sender chooses a transmission power (which may be pre-determined, e.g.~due to hardware limitations), and the received power decreased polynomially with the distance from the sender.  A transmission is successful if the received power from the sender is large enough to overcome the interference caused by other senders plus the background noise.

With our baseline being the SINR model, we then consider four ways of ``improving" a network: adding power control, adding interference cancellation, adding both power control and interference cancellation, and decreasing the SINR threshold.  With all of these modifications it is easy to see that the optimal capacity can only increase, but we will show that the equilibria can become worse.  Thus ``improving" a network might actually result in worse performance.

The game-theoretic setup that we use is based on~\cite{AD09} and will be formally described in Section~\ref{sec:game-theory}, but we will give an overview here.  We start with a game in which the players are the links, and the strategies depend slightly on the model but are essentially possible power settings at which to transmit.  The utilities depend on whether or not the link was successful, and whether or not it even attempted to transmit.  In a pure Nash equilibrium every player has a strategy (i.e.~power setting) and has no incentive to deviate: any other strategy would result in smaller utility.  In a mixed Nash equilibrium every link has a probability distribution over the strategies, and no link has any incentive to deviate from their distribution.  Finally, no-regret behavior is the empirical distribution of play when all players use \emph{no-regret} algorithms, which are a widely used and studied class of learning algorithms (see Section~\ref{sec:game-theory} for a formal definition).  It is reasonably easy to see that any pure Nash is a mixed Nash, and any mixed Nash is a no-regret behavior.  For all of these, the quality of the solution is the achieved capacity, i.e.~the average number of successful links.



Our first result is for interference cancellation (IC), which has been widely proposed as a practical method of increasing network performance~\cite{Andrews2005}.  The basic idea of interference cancellation is quite simple. First, the strongest interfering signal is detected and decoded. Once decoded, this signal can then be subtracted (``canceled") from the original signal. Subsequently, the next strongest interfering signal can be detected and decoded from the now ``cleaner" signal, and so on.  As long as the strongest remaining signal can be decoded in the presence of the weaker signals, this process continues until we are left with the desired transmitted signal, which can now be decoded.  This clearly can increase the capacity of the network, and even in the worst case cannot decrease it.  And yet due to bad game-theoretic interactions it might make the achieved capacity worse:

\begin{theorem} \label{thm:IC-intro}
There exists a set of links in which the \emph{best} no-regret behavior under interference cancellation achieves capacity at most $c$ times the \emph{worst} no-regret behavior without interference cancellation, for some constant $c < 1$.  However, for every set of links the worst no-regret behavior under interference cancellation achieves capacity that is at least a constant fraction of the best no-regret behavior without interference cancellation.
\end{theorem}

Thus IC can make the achieved capacity worse, but only by a constant factor.  Note that since every Nash equilibrium (mixed or pure) is also no-regret, this implies the equivalent statements for those type of equilibria as well.  In this result (as in most of our examples) we only show a small network ($4$ links) with no background noise, but these are both only for simplicity -- it is easy to incorporate constant noise, and the small network can be repeated at sufficient distance to get examples with an arbitrarily large number of links.

We next consider power control (PC), where senders can choose not just whether to transmit, but at what power to transmit.  It turns out that any equilibrium without power control is also an equilibrium with power control, and thus we cannot hope to find an example where the best equilibrium with power control is worse than the worst equilibrium without power control (as we did with IC).  Instead, we show that adding power control can create worse equilibria:
\begin{theorem} \label{thm:PC-intro}
There exists a set of links in which there is a pure Nash equilibrium with power control of value at most $c$ times the value of the worst no-regret behavior without power control, for some constant $c < 1$.  However, for every set of links the worst no-regret behavior with power control has value that is at least a constant fraction of the value of the \emph{best} no-regret behavior without power control.
\end{theorem}

Note that the first part of the theorem implies that not only is there a pure Nash with low-value (with power control), there are also mixed Nash and no-regret behaviors with low value (since any pure Nash is also mixed and no-regret).  Similarly, the second part of the theorem gives a bound on the gap between the worst and the best mixed Nashes, and the worst and the best pure Nashes.

Our third set of results is on the combination of power control and interference cancellation.  It turns out that the combination of the two can be quite harmful.  When compared to either the vanilla setting (no interference cancellation or power control) or the presence of power control without interference cancellation, the combination of IC and PC acts essentially as in Theorem~\ref{thm:PC-intro}: pure Nash equilibria are created that are worse than the previous worst no-regret behavior, but this can only be by a constant factor.  On the other hand, this factor can be super-constant when compared to equilibria that only use interference cancellation.
Let  $\Delta$ be the ratio of the length of the longest link to the minimum distance between any sender and any receiver. \footnote{Note that this definition is slightly different than the one used by \cite{HW09,Asgeirsson11,HM12} and is a bit more similar to the definition used by \cite{AD09,Din10}. The interested reader can see that this is in fact the appropriate definition in the IC setting, namely, in a setting where a receiver can decode multiple (interfering) stations.}
\begin{theorem}
There exists a set of links in which the worst pure Nash with both PC and IC (and thus the worst mixed Nash or no-regret behavior) has value at most $O(1/\log \Delta)$ times the value of the worst no-regret behavior with just IC.  However, for every set of links the worst no-regret behavior with both PC and IC has value at least $\Omega(1/\log \Delta)$ times the value of the best no-regret behavior with just IC.
\end{theorem}

This theorem means that interference cancellation ``changes the game": if interference control were not an option then power control can only hurt the equilibria by a constant amount (from Theorem~\ref{thm:PC-intro}), but if we assume that interference cancellation is present then adding power control can hurt us by $\Omega(\log \Delta)$.  Thus when deciding whether to use both power control and interference cancellation, one must be particularly careful to analyze how they act in combination.

Finally, we consider the effect of decreasing the SINR threshold $\beta$ (this value will be formally described in Section~\ref{sec:models}).  We show that, as with IC, there are networks in which a decrease in the SINR threshold can lead to \emph{every} equilibrium being worse than even the worst equilibrium at the higher threshold, despite the capacity increasing or staying the same:
\begin{theorem}
There exists a set of links and constants $1 < \beta' < \beta$ in which the best no-regret behavior under threshold $\beta'$ has value at most $c$ times the value of the worst no-regret behavior under threshold $\beta$, for some constant $c < 1$.  However, for any set of links and any $1 < \beta' < \beta$ the value of the worst no-regret behavior under $\beta'$ is at least a constant fraction of the value of the best no-regret behavior under $\beta$.
\end{theorem}
Our main network constructions illustrating Braess's paradox in the studied settings are summarized in Fig. \ref{fig:nash_all}.
\begin{figure}[ht!]
\begin{center}
\includegraphics[scale=0.2]{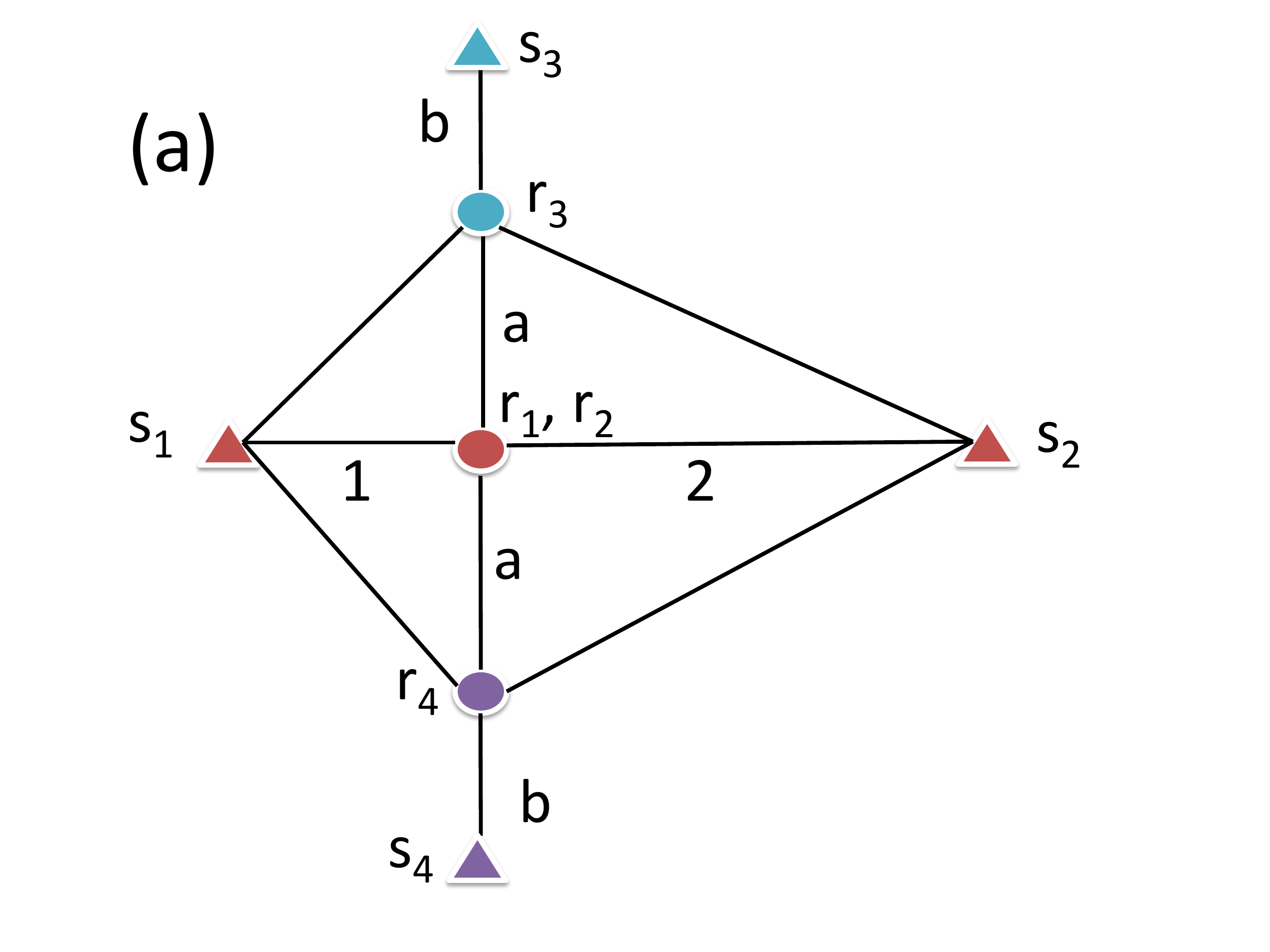}
\hfill
\includegraphics[scale=0.2]{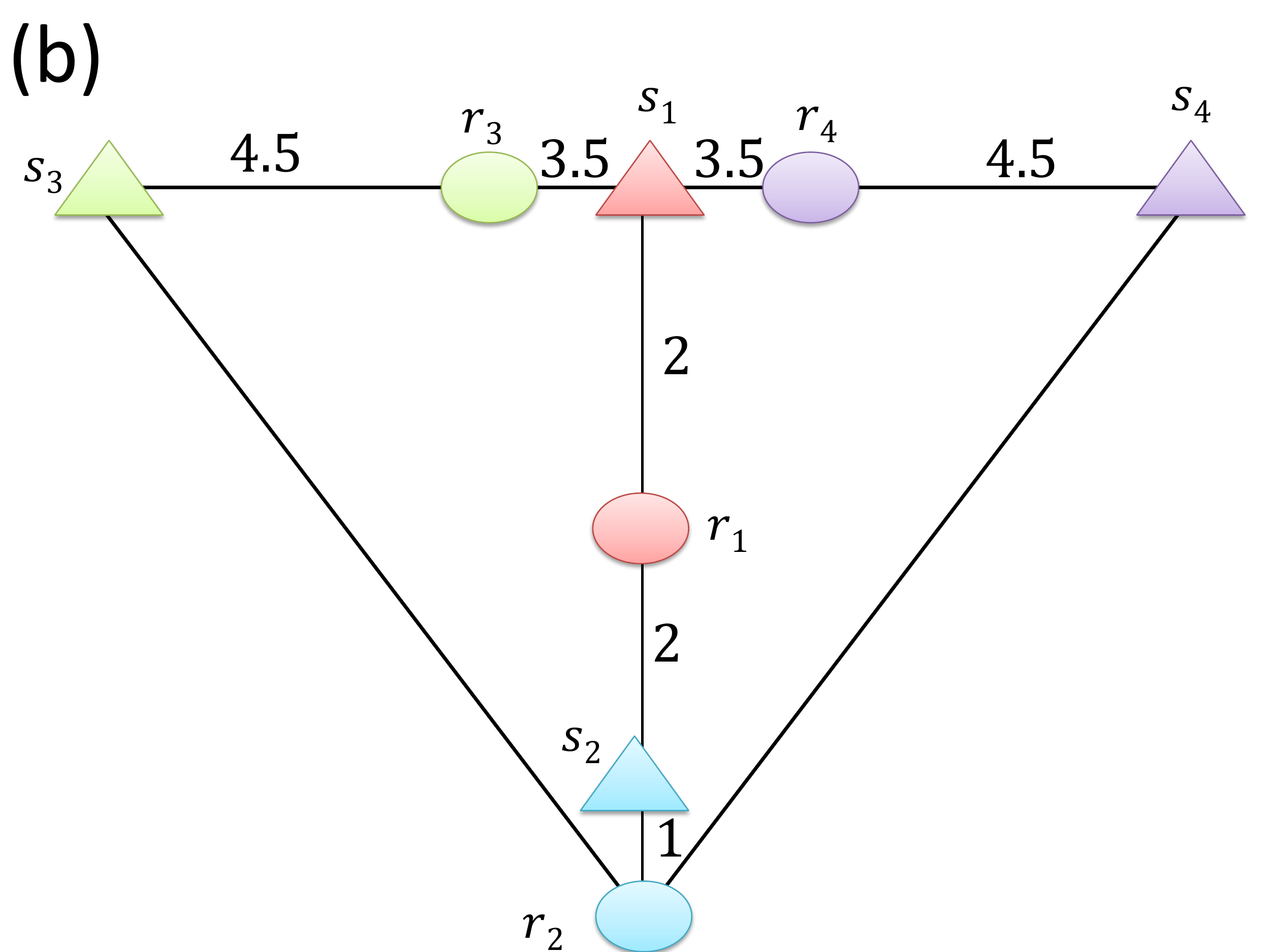}
\hfill
\includegraphics[scale=0.2]{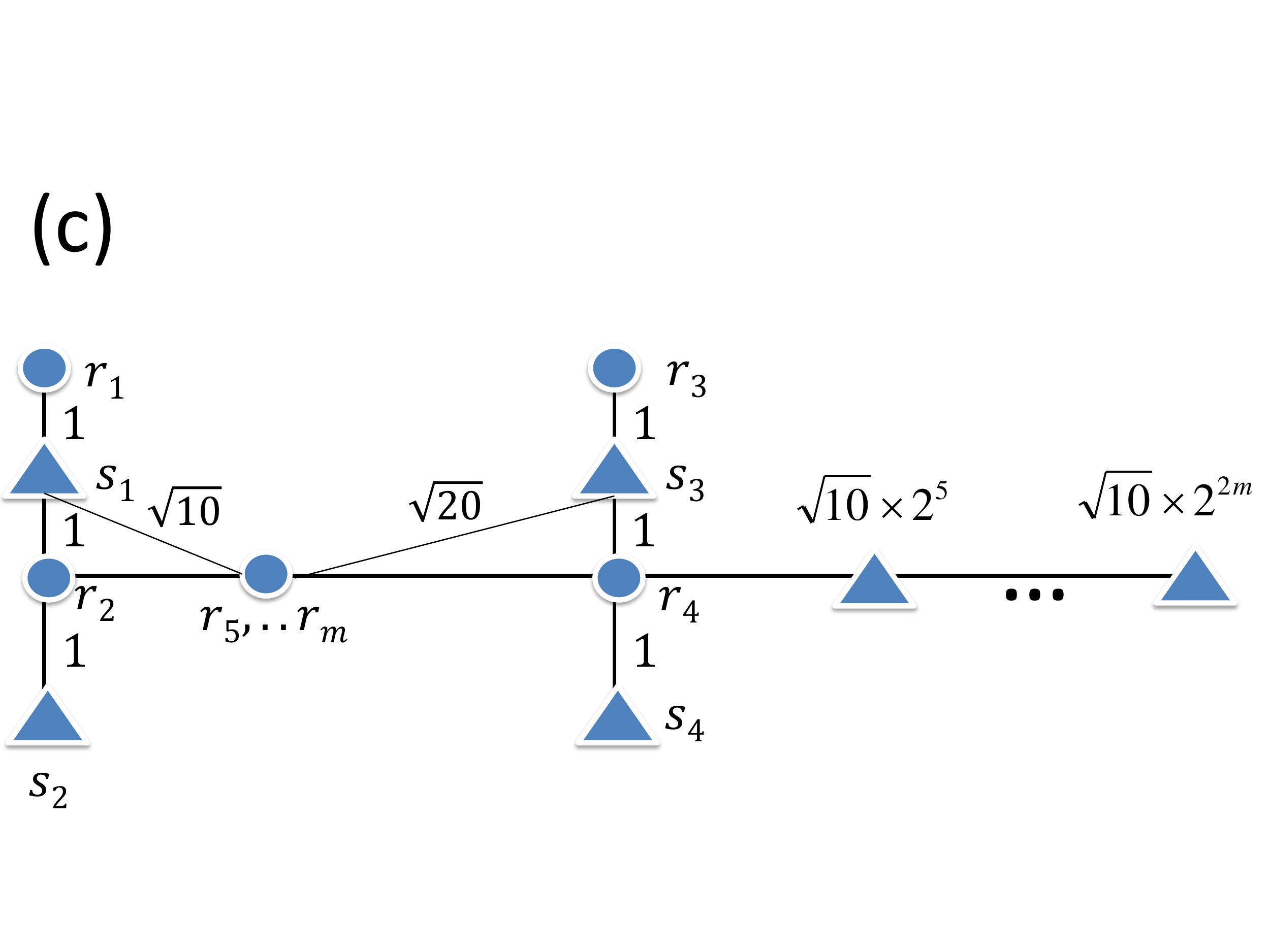}
\hfill
\includegraphics[scale=0.2]{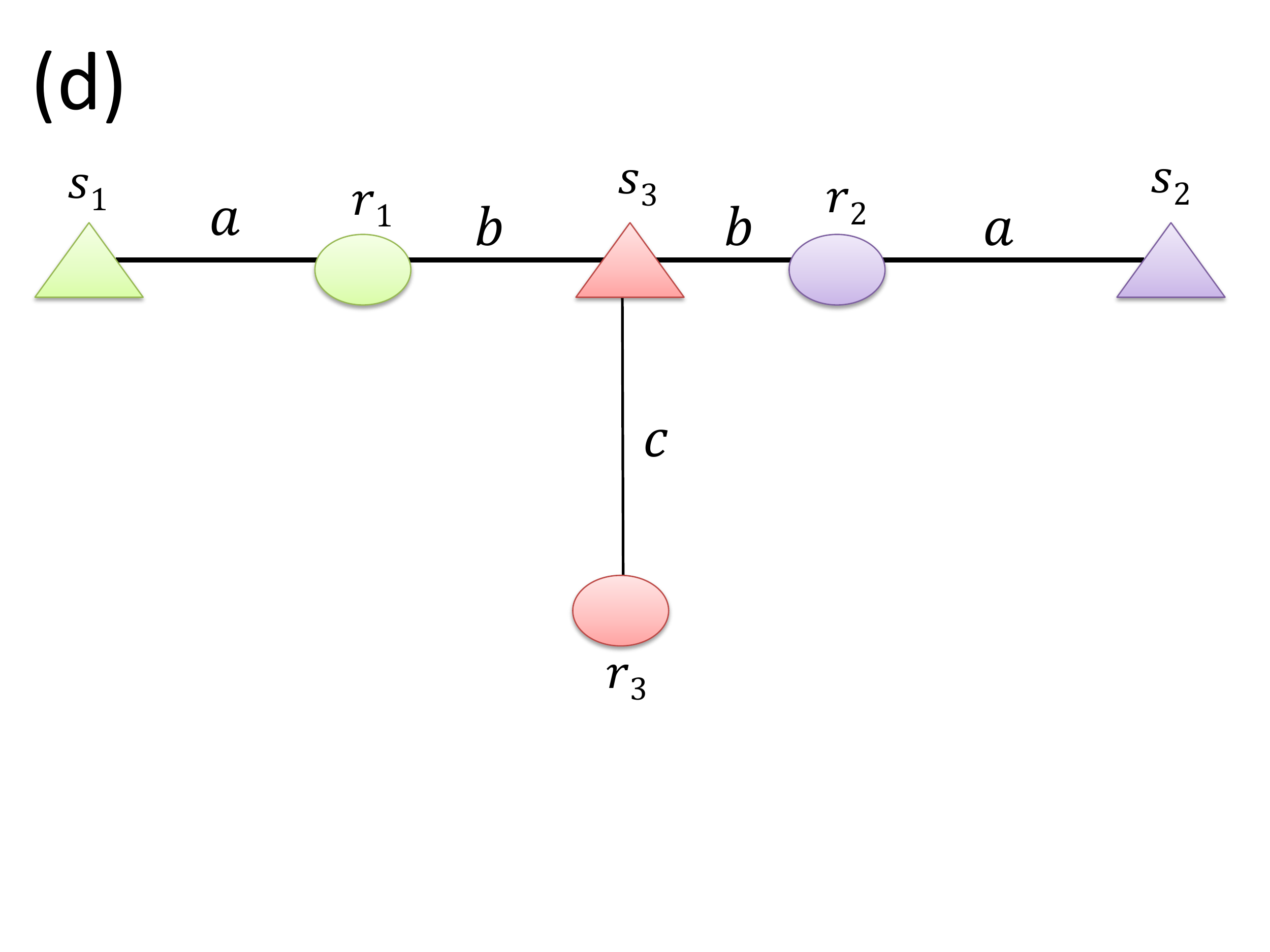}
\caption{ \label{fig:nash_all}
\sf
Schematic illustration of the main lower bounds illustrating the Braess's paradox with
(a) IC: a network in which every no-regret behavior without IC is better than any no-regret behavior solution with IC;
(b) PC: a network in which there exists a pure Nash equilibrium with PC which is worse than any no-regret behavior with IC;
(c) PIC: a network with a pure Nash equilibrium in the PIC setting which is $\Omega(\log \Delta)$ worse than any no-regret behavior in the IC setting but \emph{without} power control; and
(d) Decreased SINR threshold $\beta'<\beta$: a network in which every no-regret behavior with $\beta'$ has a smaller value than any no-regret behavior with higher SINR threshold $\beta$.
Edge weights represent distances.
}
\end{center}
\end{figure}


\subsection{Related Work}

The capacity of \emph{random} networks was examined in the seminal paper of Gupta and Kumar~\cite{GuKu00}, who proved tight bounds in a variety of models.  But only recently has there been a significant amount of work on algorithms for determining the capacity of \emph{arbitrary} networks, particularly in the SINR model.  This line of work began with Goussevskaia, Oswald, and Wattenhofer~\cite{GoOsWa07}, who gave an $O(\log \Delta)$-approximation for the uniform power setting (i.e.~the vanilla model we consider).  Goussevskaia, Halld\'orson, Wattenhofer, and Welzl~\cite{GHWW09} then improved this to an $O(1)$-approximation (still under uniform powers), while Andrews and Dinitz~\cite{AD09} gave a similar $O(\log \Delta)$-approximation algorithm for the power control setting.  This line of research was essentially completed by an $O(1)$-approximation for the power control setting due to Kesselheim~\cite{K11}.

In parallel to the work on approximation algorithms, there has been some work on using game theory (and in particular the games used in this paper) to help design distributed approximation algorithms.  This was begun by Andrews and Dinitz~\cite{AD09}, who gave an upper bound of $O(\Delta^{2\alpha})$ on the price of anarchy for the basic game defined in Section~\ref{sec:game-theory}.  But since computing the Nash equilibrium of a game is PPAD-complete~\cite{DGP06}, we do not expect games to necessarily converge to a Nash equilibrium in polynomial time.  Thus Dinitz~\cite{Din10} strengthened the result by showing the same upper bound of $O(\Delta^{2\alpha})$ for no-regret behavior.  This gave the first distributed algorithm with a nontrivial approximation ratio, simply by having every player use a no-regret algorithm.  The analysis of the same game was then improved to $O(\log \Delta)$ by \'Asgeirsson and Mitra~\cite{Asgeirsson11}.

There is very little work on interference cancellation in arbitrary networks from an algorithmic point of view, although it has been studied quite well from an information-theoretic point of view (see e.g.~\cite{ETW08,CE87}).  Recently Avin et al.~\cite{AvinCHKLPP12} studied the topology of \emph{SINR diagrams} under interference cancellation, which are a generalization of the SINR diagrams introduced by Avin et al.~\cite{Avin2009PODC} and further studied by Kantor et al.~\cite{KLPP11} for the SINR model without interference cancellation.  These diagrams specify the reception zones of transmitters in the SINR model, which turn out to have several interesting topological and geometric properties but have not led to a better understanding of the fundamental capacity question.

\section{Preliminaries}

\subsection{The Communication Model} \label{sec:models}
We model a wireless network as a set of links $L=\{\ell_1, \ell_2, \ldots, \ell_n\}$ in the plane, where each link $\ell_v=(s_v, r_v)$ represents a communication request from a sender $s_v$ to a receiver $r_v$. The $n$ senders and receivers are given as points in the Euclidean plane. The Euclidean distance between two points $p$ and $q$ is denoted $d(p,q)$.
The distance between sender $s_i$ and receiver $r_j$ is denoted by $d_{i,j}$. We adopt the physical model (sometime called the SINR model) where the received signal strength of transmitter $s_i$ at the receiver $r_j$ decays with the distance and it is given by $P_{i,j}=P_i/d_{i,j}^\alpha$, where $P_i \in [0, P_{\max}]$ is the transmission power of sender $s_i$ and $\alpha >0$ is the path-loss exponent.  Receiver $r_j$ successfully receives a message from sender $s_j$ iff $\SINR_{j}(L) ~=~ \frac{P_{j,j}}{\sum_{\ell_i \in L} P_{i,j}+\Noise} \geq \beta$, where $\Noise$ denotes the amount of background noise and $\beta>1$ denotes the minimum SINR required for a message to be successfully received.
The \emph{total interference} that receiver $r_j$ suffers from the set of links $L$ is given by $\sum_{i \neq j} P_{i,j}$.
Throughout, we assume that all distances $d_{i,j}$ are normalized so that $\min_{s_i, r_j}d_{i,j} = 1$ hence the maximal link length is $\Delta$, i.e., $\Delta=\max_{i} d_{i,i}$ and any received signal strength $P_{i,j}$ is bounded by $P_{i,j} \leq P_{\max}$.

In the vanilla SINR model we require that $P_i$ is either $0$ or $P_{\max}$ for every transmitter.  This is sometimes referred to as \emph{uniform} powers.  When we have power control, we allow $P_i$ to be any integer in $[0, P_{\max}]$.

Interference cancellation allows receivers to cancel signals that they can decode.  Consider link $\ell_j$.  If $r_j$ can decode the signal with the largest received signal, then it can decode it and remove it.  It can repeat this process until it decodes its desire message from $s_j$, unless at some point it gets stuck and cannot decode the strongest signal.  Formally, $r_j$ can decode $s_i$ if $
P_{i,j} / (\sum_{\ell_k : P_{k,j} < P_{i,j}} P_{k,j}+\Noise) \geq \beta$ (i.e.~it can decode $s_i$ in the presence of weaker signals) and if it can decode $s_k$ for all links $\ell_k$ with $P_{k,j} > P_{i,j}$.  Link $\ell_j$ is successful if $r_j$ can decode $s_j$.

The following key notion, which was introduced in~\cite{HW09} and extended to arbitrary powers by \cite{Kesselheim}, plays an important role in our analysis. The \emph{affectance} $a_{w}^P(v)$ of link $\ell_v$ caused by another link $\ell_w$ with a given power assignment vector $P$ is defines to be $a_{w}^P(v)=\min\left\{1, c_v(\beta) P_{w,v}/P_{v,v} \right\}$,
where $c_v(\beta)=\beta/(1-\beta \Noise d_{v,v}^{\alpha}/P_v)$.
Informally, $a_{w}^P(v)$ indicates the amount of (normalized) interference that link $\ell_w$ causes at link $\ell_v$.  It is easy to verify that link $\ell_v$ is successful if and only if $\sum_{w \neq v} a_w^P(v) \leq 1$.

When the powers of $P$ are the same $P_u=P_{\max}$ for every $\ell_u$, (i.e.~uniform powers), we may omit it and simply write $a_{w}(v)$.  For a set of links $L$ and a link $\ell_w$, the total affectance caused by $\ell_w$ is $a_w(L) = \sum_{\ell_v \in L} a_w(v)$. In the same manner, the total affectance caused by $L$ on the link $\ell_w$ is $a_L(w) = \sum_{\ell_v \in L} a_v(w)$.
We say that a set of links $L$ is \emph{$\beta$-feasible} if $\SINR_{v}(L)\geq \beta$ for all $\ell_v \in L$, i.e.~every link achieves SINR above the threshold (and is thus successful even without interference cancellation).  It is easy to verify that $L$ is $\beta$-feasible if and only if $\sum_{\ell_w \in L} a_{w}(v)\leq 1$ for every $\ell_v \in L$.

Following \cite{HM12}, we say that a link set $L$ is \emph{amenable} if the total affectance caused by any single link is bounded by some constant, i.e.,  $a_{u}(L) =O(1)$ for every $u \in L$.
The following basic properties of amenable sets play an important role in our analysis.
\begin{fact}\cite{Asgeirsson11}
\label{fc:amenable}
(a) Every feasible set $L$ contains a subset $L' \subseteq L$, such that $L'$ is amenable and $|L'|\geq |L|/2$.\\
(b) For every amenable set $L'$ that is $\beta$-feasible with uniform powers, for every other link $u$, it holds that $\sum_{v \in L'}a_{u}(v)=O(1)$.
\end{fact}


\subsection{Basic Game Theory} \label{sec:game-theory}
We will use a game that is essentially the same as the game of Andrews and Dinitz~\cite{AD09}, modified only to account for the different models we consider.  Each link $\ell_i$ is a player with $P_{\max}+1$ possible strategies: broadcast at power $0$, or at integer power $P \in \{1, \dots, P_{\max}\}$. A link has utility $1$ if it is successful, has utility $-1$ if it uses nonzero power but is unsuccessful, and has utility $0$ if it does not transmit (i.e.~chooses power $0$).  Note that if power control is not available, this game only has two strategies: power $0$ and power $P_{\max}$. Let $S$ denote the set of possible strategies.  A \emph{strategy profile} is a vector in $S^n$, where the $i$'th component is the strategy played by link $\ell_i$.  For each link $\ell_i$, let $f_i: S^n \rightarrow \{-1, 0, 1\}$ be the function mapping strategy profiles to utility for link $i$ as described.  Given a strategy profile $a$, let $a_{-i}$ denote the profile without the $i$'th component, and given some strategy $s \in S$ let $f_i(s, a_{-i})$ denote the utility of $\ell_i$ if it uses strategy $s$ and all other links use their strategies from $a$.

A pure Nash equilibrium is a strategy profile in which no player has any incentive to deviate from their strategy.  Formally, $a \in S^n$ is a pure Nash equilibrium if $f_i(a) \geq f_i(s, a_{-i})$ for all $s \in S$ and players $\ell_i \in L$.  In a mixed Nash equilibrium~\cite{Nash50}, every player $\ell_i$ has a probability distribution $\pi_i$ over $S$, and the requirement is that no player has any incentive to change their distribution to some $\pi'$.  So $\E[f_i(a)] \geq \E[f_i(\pi', a_{-i})]$ for all $i \in \{1, \ldots, n\}$, where the expectation is over the random strategy profile $a$ drawn from the product distribution defined by the $\pi_i$'s, and $\pi'$ is any distribution over $S$.

To define no-regret behavior, suppose that the game has been played for $T$ rounds and let $a^t$ be the realized strategy profile in round $t \in \{1, \ldots, T\}$. The \emph{history} $\mathrm{H}=\{a^1, \ldots, a^T\}$ of the game is the sequence of the $T$ strategy profiles.
The \emph{regret} $\mathcal{R}_i$ of player $i$ in an history $\mathrm{H}$ is defined to be
\begin{equation*}
\mathcal{R}_i(\mathrm{H})=\textstyle \max_{s \in S} \frac1T \sum_{t=1}^T f_i(s, a^t_{-i}) - \frac1T \sum_{t=1}^T f_i(a^t).
\end{equation*}
The regret of a player is intuitively the amount that it lost by not playing some fixed strategy.  An algorithm used by a player is known as a \emph{no-regret} algorithm if it guarantees that the regret of the player tends to $0$ as $T$ tends to infinity.
There is a large amount of work on no-regret algorithms, and it turns out that many such algorithms exist (see e.g.~\cite{ACFS03,LW94}). Thus we will analyze situations where every player has regret at most $\epsilon$, and since this tends to $0$ we will be free to assume that $\epsilon$ is arbitrarily small, say at most $1/n$.  Clearly playing a pure or mixed Nash is a no-regret algorithm (since the fact that no one has incentive to switch to any other strategy guarantees that each player will have regret $0$ in the long run), so analyzing the worst or best history with regret at most $\epsilon$ is more general than analyzing the worst or best mixed or pure Nash.  We will call a history in which all players have regret at most $\epsilon$ an $\epsilon$-regret history. Formally an history $\mathrm{H}=\{a^1, \ldots, a^T\}$ is an $\epsilon$-\emph{regret} history if
$\mathcal{R}_i(\mathrm{H})\leq \epsilon$ for every player $i \in \{1, \ldots, n\}$.

A simple but important lemma introduced in~\cite{Din10} and used again in~\cite{Asgeirsson11} relates the average number of \emph{attempted} transmissions to the average number of \emph{successful} transmissions.  Fix an $\epsilon$-regret history, let $s_u$ be the fraction of times in which $u$ successfully transmitted, and let $p_u$ be the fraction of times in which $u$ attempted to transmit.  Note that the average number of successful transmissions in a time slot is exactly $\sum_u s_u$, so it is this quantity which we will typically attempt to bound.  The following lemma lets us get away with bounding the number of attempts instead.

\begin{lemma}[\cite{Din10}] \label{lem:attempts}
$\sum_u s_u \leq \sum_u p_u \leq 2 \sum_u s_u + \epsilon n \leq O(\sum_u s_u)$~.
\end{lemma}


\paragraph{Notation:}
Let $L$ be a fixed set of $n$ links embedded in $\R^2$.  Let $\NoRegret_{\min}(L)$ denote the minimum number of successful links (averaged over time) in any $\epsilon$-regret history, and similarly let $\NoRegret_{\max}(L)$ denote the maximum number of successful links (averaged over time) in any $\epsilon$-regret history.  Define $\NoRegret_{\max}^{IC}(L)$, and $\NoRegret_{\min}^{IC}(L)$ similarly for the IC setting, $\NoRegret_{\max}^{PC}(L)$ and $\NoRegret_{\min}^{PC}(L)$ for the PC setting, and $\NoRegret_{\max}^{PIC}(L)$ and $\NoRegret_{\min}^{PIC}(L)$ for the setting with both PC and IC.  Finally, let $\NoRegret_{\max}^{\beta}(L)$, $\NoRegret_{\min}^{\beta}(L)$ be for the corresponding values for the vanilla model when the SINR threshold is set to $\beta$ (this is hidden in the previous models, but we pull it out in order to compare the effect of modifying $\beta$).

While we will focus on comparing the equilibria of games utilizing different wireless technologies, much of the previous work on these games instead focuses on a single game and analyzes its equilibria with respect to OPT, the maximum achievable capacity.  The \emph{price of anarchy} (PoA) is the ratio of OPT to the value of the worst mixed Nash~\cite{KP99}, and the \emph{price of total anarchy} (PoTA) is the ratio of OPT to the value of the worst $\epsilon$-regret history~\cite{BHLR08}.  Clearly PoA $\leq$ PoTA.  While it is not our focus, we will prove some bounds on these values as corollaries of our main results.

\section{Interference Cancellation}

We begin by analyzing the effect on the equilibria of adding interference cancellation.  We would expect that using IC would result equilibria with larger values, since the capacity of the network might go up (and can certainly not go down).   We show that this is not always the case: there are sets of links for which even the \emph{best} $\epsilon$-regret history using IC is a constant factor worse than the \emph{worst} $\epsilon$-regret history without using IC.

\begin{theorem} \label{thm:IC_lower}
There exists a set of links $L$ such that
$\NoRegret_{\max}^{IC}(L)\leq \NoRegret_{\min}(L)/c$ for some constant $c > 1$.
\end{theorem}
\begin{proof}
Let $L'$ be the four link network depicted in Figure~\ref{fig:nash_all}(a), with $b = 3/2$ and $a = \sqrt{8.8} - b$.  We will assume that the threshold $\beta$ is equal to $1.1$, the path-loss exponent $\alpha$ is equal to $2$, and the background noise $\Noise = 0$ (none of these are crucial, but make the analysis simpler).  Let us first consider what happens without using interference cancellation.  Suppose each link has at most $\epsilon$-regret, and for link $\ell_i$ let $p_i$ denote the fraction of times at which $s_i$ attempted to transmit.  It is easy to see that link $1$ will always be successful, since the received signal strength at $r_1$ is $1$ while the total interference is at most $(1/4) + 2(1/8.8) = 1/\beta$.  Since $\ell_1$ has at most $\epsilon$-regret, this implies that $p_1 \geq 1- \epsilon$.

On the other hand, whenever $s_1$ transmits  it is clear that link $\ell_2$ cannot be successful, as its SINR is at most $1/4$.  So if $s_2$ transmitted every time it would have average utility at most $-(1-\epsilon) + \epsilon = -1 + 2\epsilon < 0$ (since $\epsilon < 1/2$), while if it never transmitted it would have average utility $0$.  Thus its average utility is at least $-\epsilon$.  Since it can succeed only an $\epsilon$ fraction of the time (when link $1$ is not transmitting), we have that $\epsilon - (p_2 - \epsilon) \geq -\epsilon$ and thus $p_2 \leq 3\epsilon$.
Since the utility of $s_2$ is at least $-\epsilon$, it holds that the fraction of times at which both $s_1$ and $s_2$ are transmitting is at most $2\epsilon$.

Now consider link $\ell_3$.  If links $\ell_1$ and $\ell_2$ both transmit, then $\ell_3$ will fail since the received SINR will be at most $\frac{1/b^2}{(1/(1+a^2)) + (1/(4+a^2))} \approx 0.92 < 1.1$.  On the other hand, as long as link $\ell_2$ does not transmit then $\ell_3$ will be successful, as it will have SINR at least $\frac{1/b^2}{(1/(1+a^2)) + (1/(2a+b)^2)}\geq 1.2 > 1.1$.  Thus by transmitting at every time step $\ell_3$ would have average utility at least $(1-2\epsilon) - 2\epsilon = 1-4\epsilon > 0$ (since $\epsilon < 1/4$), and thus we know that $\ell_3$ gets average utility of at least $1-5\epsilon$, and thus successfully transmits at least $1-5\epsilon$ fraction of the times.  $\ell_4$ is the same by symmetry.  Thus the total value of any history in which all links have regret at most $\epsilon$ is at least $\NoRegret_{\min}(L) \geq 1-\epsilon  + 2(1-5\epsilon) = 3 - 11\epsilon$.

Let us now analyze what happens when using interference cancellation and bound $\NoRegret^{IC}_{\max}(L)$.
Suppose each link has at most $\epsilon$-regret, and for link $\ell_i$ let $q_i$ denote the fraction of times at which $s_i$ attempted to transmit.  As before, $\ell_1$ can always successfully transmit and thus does so in at least $1-\epsilon$ fraction of times.  But now, by using interference cancellation it turns out that $\ell_2$ can also always succeed.  This is because $r_2$ can first decode the transmission from $s_1$ and cancel it, leaving a remaining SINR of at least $\frac{1/4}{2/(a+b)^2} = \beta$.  Thus $\ell_2$ will also transmit in at least $1-\epsilon$ fraction of times and hence so far $1-\epsilon \leq q_1,q_2 \leq 1$.  Note that since $a^2+1<b^2$, it holds that
$r_3$ cannot cancel $s_1$ or $s_2$ before decoding $s_3$ (i.e., $P_{1,3},P_{2,3}<P_{3,3}$). Hence, cancellation is useless.
But now at $r_3$ the strength of $s_1$ is $1/(1+a^2)>0.317$, the strength of $s_2$ is $1/(4+a^2)> 0.162$, and the strength of $s_3$ is $1/b^2 =4/9$.  Thus $r_3$ cannot decode any messages when $s_1, s_2$, and $s_3$ are all transmitting since its SINR is at most $0.92<\beta$, which implies that $\ell_3$ can only succeed on at most $2\epsilon$ fraction of times. The link $\ell_4$ is the same as the link $\ell_3$ by symmetry.  Thus the total value of any history in which all links have an $\epsilon$-regret is at most $\NoRegret^{IC}_{\max}(L) \leq 2 + 4\epsilon$. Thus $\NoRegret_{\min}(L)/\NoRegret^{IC}_{\max}(L)\geq 3/2-o(1)$ as required.
\QED \end{proof}
It turns out that no-regret behavior with interference cancellation cannot be much worse than no-regret behavior without interference cancellation -- as in Braess's paradox, it can only be worse by a constant factor.
\begin{theorem}
\label{thm:nash_ic_upper}
$\NoRegret^{IC}_{\min}(L)\geq \NoRegret_{\max}(L)/c$
for any set of links $L$ and some constant $c \geq 1$.
\end{theorem}
\begin{proof}
Consider an $\epsilon$-regret history without IC that maximizes the average number of successful links, i.e.~one that achieves $\NoRegret_{\max}(L)$ value.  Let $p_i$ denote the fraction of times at which $s_i$ attempted to transmit in this history, so $\sum_{i \in L} p_i = \Theta(\NoRegret_{\max}(L))$ by Lemma~\ref{lem:attempts}.  Similarly, let $q_i$ denote that fraction of times at which $s_i$ attempted to transmit in an $\epsilon$-regret history \emph{with} IC that achieves value of only $\NoRegret^{IC}_{\min}(L)$, and so $\sum_{i \in L} q_i = \Theta(\NoRegret^{IC}_{\min}(L))$.
%
%


Note that since the best average number of successful connections in the non-IC case is $\NoRegret_{\max}(L)$, there must exist some set of connections $A \subseteq L$ such that $|A| \geq \NoRegret_{\max}(L)$ and $A$ is feasible without IC.  Let $B = \{i : q_i \geq 1/2\}$ and let $A' = A \setminus B$.  If $|B \cap A| \geq |A| / 2$ then we are done, since then
\begin{equation*}
\textstyle \NoRegret_{\max}(L) \leq |A| \leq 2|B| \leq 4 \sum_{\ell_i \in L} q_i=4 \cdot \NoRegret^{IC}_{\min}(L)
\end{equation*}
as required.
So without loss of generality we will assume that $|B \cap A| < |A|/2$, and thus that $|A'| > |A|/2$.  Note that $A'$ is a subset of $A$, and so it is feasible in the non-IC setting.

Now let $\widehat{A} = \{i \in A' : \sum_{j \in A'} a_i(j) \leq 2\}$ be an amenable subset of $A'$.  By Fact \ref{fc:amenable}(a), it holds that $\widehat{A} \geq |A'|/ 2 \geq |A| / 4$.  Fact \ref{fc:amenable}(b) then implies that for any link $i \in L$, its total affectance on $A$ is small: $\sum_{j \in A} a_i(j) \leq c'$ for some constant $c' \geq 0$.  Thus we have that
\begin{equation} \label{eq:ic1}
\textstyle \sum_{i \in L} \sum_{j \in \widehat{A}} q_i a_i(j) \leq c' \cdot \left(\sum_{i \in L} q_i\right).
\end{equation}

On the other hand, we know that the $q_i$ values correspond to the worst history in which every link has regret at most $\epsilon$ (in the IC setting).  Let $j \in A'$.  Then $q_j < 1/2$, which means the average utility of link $\ell_j$ is at most $1/2$. Let $y_j$ be the fraction of time $s_j$ would have succeeded had it transmitted in every round. Since the average utility of the best single action is at most $1/2+\epsilon$ it holds that $y_j-(1-y_j)\leq 1/2+\epsilon$ or the that $y_j \leq \frac34 + \frac{\epsilon}{2}$.
In other words, in at least $1-y_j=\frac14 - \frac{\epsilon}{2}$ fraction of the rounds the affectance of the other links on the link $\ell_j$ must be at least $1$ (or else $j$ could succeed in those rounds even without using IC).  Thus the expected affectance (taken over a random choice of time slot) on $\ell_j$ is at least $\sum_{i \in L} a_i(j) q_i \geq \frac14 - \frac{\epsilon}{2}$.  Summing over all $j \in \widehat{A}$, we get that
\begin{equation} \label{eq:ic2}
\textstyle \sum_{j \in \widehat{A}} \sum_{i \in L} a_i(j) q_i \geq \sum_{j \in A} \frac{1-2\epsilon}{4} \geq \Omega(|\widehat{A}|).
\end{equation}

Combining equations~(\ref{eq:ic1}) and~(\ref{eq:ic2}) (and switching the order of summations) implies that $|\widehat{A}| \leq O(\sum_{i \in L} q_i)$.  Since $|\widehat{A}| \geq |A|/4 \geq \Omega(N_{\max}(L)) \geq \Omega(\sum_{i \in L} p_i)$, we get that $\sum_i p_i \leq O(\sum_{i \in L} q_i)$ as desired.
\QED \end{proof}

As a simple corollary, we will show that this lets us bound the price of total anarchy in the interference cancellation model (which, as far as we know, has not previously been bounded).  Let $OPT \subseteq L$ denote some optimal solution without IC, i.e., the set of transmitters forming a maximum $\beta$-feasible set, and let $OPT^{IC} \subseteq L$ denote some optimal solution with IC. 
\begin{corollary}
\label{cor:nash_ic_approx}
For every set of links $L$ it holds that the price of total anarchy with IC is $O(\log \Delta)$, or $|OPT^{IC}| / \mathcal{N}^{IC}_{\min}(L)=O(\log \Delta)$.
\end{corollary}
\def\APENNEDNASHICCOR{
\begin{proof}
We begin by providing the following auxiliary lemma which demonstrates that any feasible set of links under the setting of IC and power control in the range $[1, P_{\max}]$ contains a non-negligible subset of links that are feasible with the same transmission powers \emph{without} IC.
\begin{lemma}
\label{lem:ic_no_ic_feas}
For every feasible set $L$ with IC, there exists a subset $L' \subseteq L$ such that $L$ is feasible without IC and satisfies $|L'|>|L|/O \left(\log_{\beta} \left(\Delta^{\alpha} \cdot P_{\max}\right) \right)$
\end{lemma}
\begin{proof}
Let $L$ be a feasible set with IC and transmission power in $[1,P_{\max}]$. We first claim that every link $\ell_v \in L$ has canceled a subset of at most $x=\lceil \log_{\beta} \left(\Delta^{\alpha} \cdot P_{\max}\right)-1 \rceil$ links in $L$ (where the last cancellation corresponds to the signal of the designated transmitter $s_v$).
Assume towards contradiction that there exists $\ell_v \in L$ that cancels $y \geq x+1$ distinct signals transmitted by $s_1, \ldots, s_{y}$, where $s_{y}=s_{v}$ is the last canceled signal.

It then holds that $P_{i, v} \geq \beta \cdot P_{i+1,v}$ for every $i \in \{1, \ldots, y\}$, hence $P_{1, v} \geq \beta^{y} P_{v,v}$.

Since the maximum power is $P_{\max}$ and the minimal transmitter-receiver distance $d(s_i,r_j)$ is at least $1$, we have that
$P_{\max} \geq P_{1,v} \geq \beta^{y} \cdot P_{v,v} \geq \beta^{y}/\Delta^{\alpha}$ where the last inequality follows by the fact that the maximal link length (after the normalization) is $\Delta$ and the minimum power level is $1$.
Hence $y \leq \log_{\beta}(\Delta^{\alpha} \cdot P_{\max})-1$, giving a contradiction. Thus every link has cancelled at most $x$ links.
It then holds that $P_{i, v} \geq \beta \cdot P_{i+1,v}$ for every $i \in \{1, \ldots, y\}$, hence $P_{1, v} \geq \beta^{y} P_{v,v}$.  Combining this with the fact (due to normalizing distances so the smallest is equal to $1$) that every received power is at least $1/\Delta^{\alpha}$, we have that $P_{\max} \geq P_{1,v} \geq \beta^{y} \cdot P_{v,v} \geq \beta^{y}/\Delta^{\alpha}$.  Hence $y \leq \log_{\beta}(\Delta^{\alpha} \cdot P_{\max})-1$, giving a contradiction. Thus every link has cancelled at most $x$ links.

The subset $L'$ of feasible links without IC is created as follows. Initially set $Q \gets L$. Until $Q$ is non-empty, consider some link $\ell_i \in Q$ and let $M_i$ be the subset of links cancelled by receiver $r_i$ (other than $\ell_i$ itself). Add $\ell_i$ to $L'$ and set $Q=Q \setminus M_{i}$. Since for every link $\ell_i$ in $L'$ there are at most $x$ links in $L \setminus L'$, it holds that $|L'| \geq |L|/x$. Moreover, it follows by construction that $L'$ is feasible without IC and using the same transmission powers as before. The lemma follows.
\QED \end{proof}

We proceed by showing that $|OPT^{IC}|/\mathcal{N}^{IC}_{\min}(L))=O(\log \Delta)$.
According to Theorem 2 of \cite{Asgeirsson11}, it holds that
$\mathcal{N}_{\min}(L)\geq |OPT|/c'$ for some constant $c'$.
Hence, by combining with Theorem~\ref{thm:nash_ic_upper} we get that
\begin{equation*}
\mathcal{N}^{IC}_{\min}(L) \geq \mathcal{N}_{\min}(L)/c \geq |OPT|/(c \cdot c') \geq |OPT^{IC}|/O(\log \Delta),
\end{equation*}
where the last inequality follows by Lemma \ref{lem:ic_no_ic_feas}.
\par Finally, we show that this analysis is actually tight by exhibiting a network where there is a bad pure Nash equilibrium, and thus there are bad no-regret histories.  Consider the $m=\lfloor (\log \Delta) \rfloor$-linkset network illustrated in Fig. \ref{fig:poa_ic}.
The transmitters $\widetilde{s}_1$ and $\widetilde{s}_2$ are
equidistant from the set of $m$-receivers $R'=\{r_1, \ldots, r_m\}$.
Since $\widetilde{s}_1$ and $\widetilde{s}_2$ are closer to $R'$ than any other transmitter $s_i$, if both $\widetilde{s}_1$ and $\widetilde{s}_2$ transmit then none of the links $\ell_i=\langle s_i, r_i\rangle$ can be satisfied. By letting the links $\widetilde{\ell}_j=\langle \widetilde{s}_j, \widetilde{r}_j\rangle$ for $j \in \{1, 2\}$ be sufficiently short, these 2 links can always succeed no matter which other links transmit.  Thus $\{\widetilde{\ell}_1, \widetilde{\ell}_2\}$ form a pure Nash equilibrium.  On the other  hand, clearly the set of links $\{\ell_i\}_{i \in [\log \Delta]}$ are feasible with interference cancellation.  Thus $|OPT^{IC}| / \mathcal{N}^{IC}_{\min}(L) \geq \Omega(\log \Delta)$.
\QED \end{proof}
\begin{figure}[h!]
\begin{center}
\includegraphics[scale=0.35]{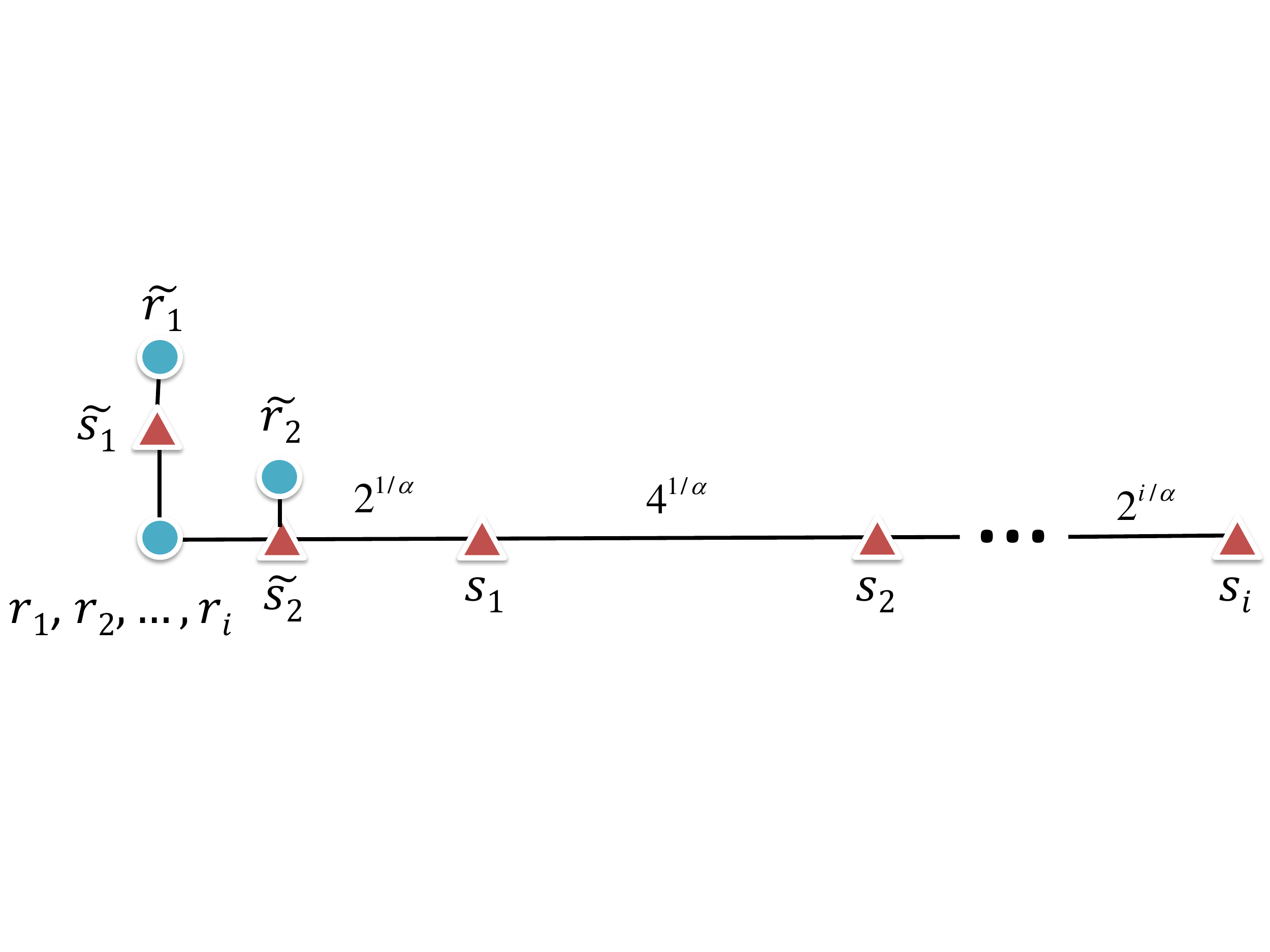}
\caption{ \label{fig:poa_ic}
\sf
Schematic illustration of a network in which the price of anarchy in the IC setting is $\Omega(\log \Delta)$.
}
\end{center}
\end{figure}
}
\APENNEDNASHICCOR

\section{Power Control}
In the power control setting, each transmitter $s_v$ either broadcasts at power $0$ or broadcasts at some arbitrary integral power level $P_v \in [1, P_{\max}]$.  Our main claim is that Braess's paradox is once again possible: there are networks in which adding power control can create worse equilibria. For illustration of such a network, see Fig. \ref{fig:nash_all}(b).
We first observe the following relation between no-regret solutions with or without power control.
\begin{observation}
\label{obs:nash_pc_contain}
Every no-regret solution in the uniform setting, is also a no-regret solution in the PC setting.
\end{observation}
Hence, we cannot expect the best no-regret solution in the PC setting to be smaller than the worst no-regret solution in the uniform setting. Yet, the paradox still holds.
\begin{theorem}
\label{thm:pmax_nash_lb}
There exists a configuration of links $L$ satisfying $\NoRegret^{PC}_{\min}(L) \leq \NoRegret_{\min}(L)/c$ for some constant $c>1$.
\end{theorem}
\def\APPENDPCEXAMP{
\begin{proof}
Since any pure-Nash solution is also a no-regret solution, in this case it is sufficient to find a pure Nash equilibrium that uses power control that is a constant factor worse than $\NoRegret_{\min}(L)$, the worst no-regret behavior without power control.  Let $L$ be the set of links as illustrated in Fig.~\ref{fig:nash_all}(b). Let $\Noise=0$, $\beta=1.1$, $\alpha=3$ and $P_{\max}=2$.
Let us first consider the case of power control, where each sender $s_v$ transmits with power $P_v \in [1,P_{\max}]$. Let $N=\{\ell_1, \ell_2\}$ and $s_1$ transmits with $P_1=P_{\max}$, $s_2$ transmits with $P_2=1$ and both other transmitters user power $0$.  It is easy to see that this is a pure Nash: the SINR of $\ell_1$ is $\frac{P_{\max}/8}{1/8} > \beta$ and the SINR of $\ell_2$ is $\frac{1/1}{P_{\max}/(5)^3} > \beta$, while even if they used power $P_{\max}$ links $\ell_3$ and $\ell_4$ would not be able to overcome the interference caused by $\ell_1$ since $P_{\max}/4.5^3< 1/3.5^3$ (as $P_{\max}=2$).

We will now analyze the worst no-regret behavior without power control.  Given a history in which all players have regret at most $\epsilon$, let $p_i$ denote the fraction of times in which $\ell_i$ broadcasts.  Note that $\ell_2$ is always feasible, since under uniform powers the SINR of $\ell_2$ is at least $1 / (P_{1,2} + P_{3,2} + P_{4,2}) \geq \beta$ where $P_{1,2} =1/125$ and $P_{3,2}=P_{4,2}=1/(64+25)^{3/2}$.  Thus $p_2 \geq 1-\epsilon$, and link $\ell_2$ succeeds at least $1-\epsilon$ fraction of the time.
Since $\beta>1$ and the interfering sender $s_2$ is at the same distance to the receiver $r_1$ as its intended sender $s_1$, it holds that $\ell_1$ cannot succeed if $s_2$ transmits. Hence, if $s_1$ transmitted every time it would have average utility at most $-(1-\epsilon) + \epsilon = -1 + 2\epsilon < 0$ (since $\epsilon < 1/2$), while if it never transmitted it would have average utility $0$.  Thus its average utility is at least $-\epsilon$.  Since it can succeed only an $\epsilon$ fraction of the time (when link $\ell_2$ is not transmitting), we have that $\epsilon - (p_1 - \epsilon) \geq -\epsilon$ and thus $p_1 \leq 3\epsilon$.

Now consider link $\ell_3$. Since $s_1$ is closer to $r_3$ than $s_3$, the link $\ell_3$ cannot succeed if $s_1$ transmits. However, it can succeed if $s_1$ is not transmitting and $s_2$ and $s_4$ are transmitting, since $P_{3,3}\geq \beta (P_{2,3}+P_{4,3})$ where $P_{3,3}=1/4.5^3, P_{2,3}=1/(16+3.5^2)^{3/2}$ and $P_{4,3}=1/11.5^3$. Thus if $s_3$ had chosen to transmit at every time it would have average utility at least $1-3\epsilon-3\epsilon = 1-6\epsilon>0$ for $\epsilon \leq 1/6$. Thus $\ell_3$ must have average utility of at least $1-7\epsilon$ and thus must succeed at least $1-7\epsilon$ fraction of the time.
Link $\ell_4$ is the same by symmetry. Thus the average number of successes in any $\epsilon$-regret history is at least $1-\epsilon + 2(1-7\epsilon) = 3-15\epsilon$, which proves the theorem.
\QED \end{proof}
}
\APPENDPCEXAMP
We now prove that (as with IC) that the paradox cannot be too bad: adding power control cannot cost us more than a constant.
The proof is very similar to that of Thm. \ref{thm:nash_ic_upper} up to some minor yet crucial modifications.
\begin{theorem}
\label{thm:nash_pc_upper}
$\NoRegret^{PC}_{\min}(L)  \geq \NoRegret_{\max}(L)/c$
for any set of links $L$ and some constant $c \geq 1$.
\end{theorem}
\def\APPENDPCUB{
\begin{proof}
Fix an arbitrary $\epsilon$-regret history without power control, where all transmitters transmit with $P_{\max}$ and let $p_i$ be the fraction of rounds in which $s_i$ attempts to transmit.  Similarly, fix an arbitrary $\epsilon$-regret history with power control, and let $q_i$ be the fraction of rounds in which $s_i$ attempts to transmit (at any nonzero power).  By Lemma~\ref{lem:attempts}, it is sufficient to prove that $\sum_{i \in L} p_i \leq O(\sum_{i \in L} q_i)$.

Note that since the average number of successful connections in the history of the uniform case is $\sum p_i$, there must exist some set of connections $A \subseteq L$ that transmitted successfully in some round $t \in [1, T]$ such that $|A| \geq \Omega(\sum_i p_i)$ and $A$ is feasible when all senders transmit with power $P_{\max}$.
Let $B = \{i : q_i \geq 1/2\}$ and let $A' = A \setminus B$.  If $|B \cap A| \geq |A| / 2$ then we are done, since then
\begin{equation*}
\sum_{i \in L} p_i \leq O(|A|) \leq O(|B|) \leq O\left(\sum_{i \in L} q_i\right)
\end{equation*}
as required.
So without loss of generality we will assume that $|B \cap A| < |A|/2$, and thus that $|A'| > |A|/2$.  Note that $A'$ is a subset of $A$, and so it is feasible in the uniform setting.
Now let $\widehat{A} = \{i \in A' : \sum_{j \in A'} a_i(j) \leq 2\}$ be an amenable subset of $A'$. By Fact \ref{fc:amenable}(a), it holds that $|\widehat{A}| \geq |A'|/ 2 \geq |A| / 4$.

We have the following.
\begin{equation} \label{eq:pc1}
\sum_{i \in L} \sum_{j \in \widehat{A}} q_i a_i(j) \leq c' \cdot \left(\sum_{i \in L} q_i\right).
\end{equation}
where $\sum_{i \in L} \sum_{j \in \widehat{A}} q_i a_i(j)$ is the average affectance of $L$ on the set $\widehat{A}$ when every transmitter that transmitted in the PC history transmitted with power $P_{\max}$.  The inequality follows by Fact \ref{fc:amenable}(b).

On the other hand, we know that the $q_i$ values correspond to an $\epsilon$-regret history in the PC setting.
Consider some $j \in \widehat{A}$.  Since $\widehat A \subseteq A' = A \setminus B$, we know that $q_j < 1/2$ and thus the average utility of link $\ell_j$ is at most $1/2$. Let $y_j$ be the fraction of time $s_j$ would have succeeded has it transmitted in every round with full power $P_{\max}$. Since the average utility of the best single action is at most $1/2+\epsilon$ it holds also that the utility of transmitting with full power is at most $1/2+\epsilon$ as well, hence $y_j-(1-y_j)\leq 1/2+\epsilon$ and so $y_j \leq \frac34 + \frac{\epsilon}{2}$.
In other words, in at least $1-y_j=\frac14 - \frac{\epsilon}{2}$ fraction of the rounds the affectance of the other links on the link $\ell_j$ must be at least $1$ (or else $j$ could succeed in those rounds) when it attempted to transmit with $P_{\max}$.  Thus the average affectance on $\ell_j$ in the PC history is at least $\frac14 - \frac{\epsilon}{2}$.  Summing over all $j \in \widehat{A}$, to get that
\begin{equation} \label{eq:pc2}
\sum_{j \in \widehat{A}} \sum_{i \in L}  q_i a_i(j) \geq \sum_{j \in \widehat{A}} \frac{1-2\epsilon}{4} \geq \Omega(|\widehat{A}|),
\end{equation}
where the first inequality follows by the fact the in the true $\epsilon$-regret history in the PC setting, the average affectance on $\ell_j$ is at least $\frac14 - \frac{\epsilon}{2}$ when $s_j$ transmit with $P_{\max}$ and all other transmitters $s_{j'}$ transmit with power at most $P_{\max}$.

Combining equations~(\ref{eq:pc1}) and~(\ref{eq:pc2}) (and switching the order of summations) implies that $|\widehat{A}| \leq O(\sum_{i \in L} q_i)$.  Since $|\widehat{A}| \geq |A|/4 \geq \Omega(\sum_{i \in L} p_i)$, we get that $\sum_i p_i \leq O(\sum_{i \in L} q_i)$ as desired.
\QED \end{proof}

}
\APPENDPCUB

\begin{corollary}
\label{cor:pmax_price}
The price of total anarchy under the power control setting with maximum transmission energy $P_{\max}$ is $\Theta(\log\Delta)$.
\end{corollary}
\def\APPENDPCPA{
\begin{proof}
The upper bound of $O(\log \Delta)$ is given by \cite{Asgeirsson11}. Let $m=\lfloor \log \Delta \rfloor$. The lower bound example is given by the nested link network described in Fig. \ref{fig:nash_pc}. Let $OPT \subseteq L$ denote the solution of the optimal solution (i.e., maximum feasible set).
According to \cite{ALP09} the link set $L \setminus \{\ell^*\}$ is feasible with exponentially increasing power level $P_i >2P_{i-1}$.  In particular, since $\beta>1$, a necessary condition for feasibility is to maintain that $P_1 > P_2 > \ldots >P_m$.
Consider an $\epsilon$-regret history in which $P^*=P_{\max}$. Since the link $\ell^*=(s^*, r^*)$ is sufficiently short, $s^*$ always succeeds with any power level $P^*$ and hence $s^*$ transmits at least $1-\epsilon$ fraction of the time. For every other link $\ell_i$, $i \in \{1, \ldots, m\}$, if $s^*$ transmits then $s_i$ would fail to transmit for every transmission power $P_i \in [1, P_{\max}]$ as the interfering sender $s^*$ is closer to $r_i$ then the intended sender $s_i$ and $s^*$ transmits with full power. Consider link $s_i$ and let $p_i$ be the fraction of time it attempted to transmit (possibly with different power levels) in the $\epsilon$-regret history. If $s_i$ would have transmit in every round using any power $P^t_i \in [1, P_{\max}]$ in round $t \in \{1, \ldots, T\}$, it would have average utility at most $-(1-\epsilon) + \epsilon = -1 + 2\epsilon < 0$ (since $\epsilon < 1/2$), while if it never transmitted it would have average utility $0$.
Thus its average utility is at least $-\epsilon$.  Since it can succeed only an $\epsilon$ fraction of the time (when link $\ell^*$ is not transmitting), we have that $\epsilon - (p_i - \epsilon) \geq -\epsilon$ and thus $p_i \leq 3\epsilon$
Overall, the average number of attempted transmission is at most $1+3  \cdot m/\epsilon \leq 2$ for $\epsilon=1/(3m)$. Since $|OPT|=m$, the lemma follows.
\QED \end{proof}
\begin{figure}[h!]
\begin{center}
\includegraphics[scale=0.35]{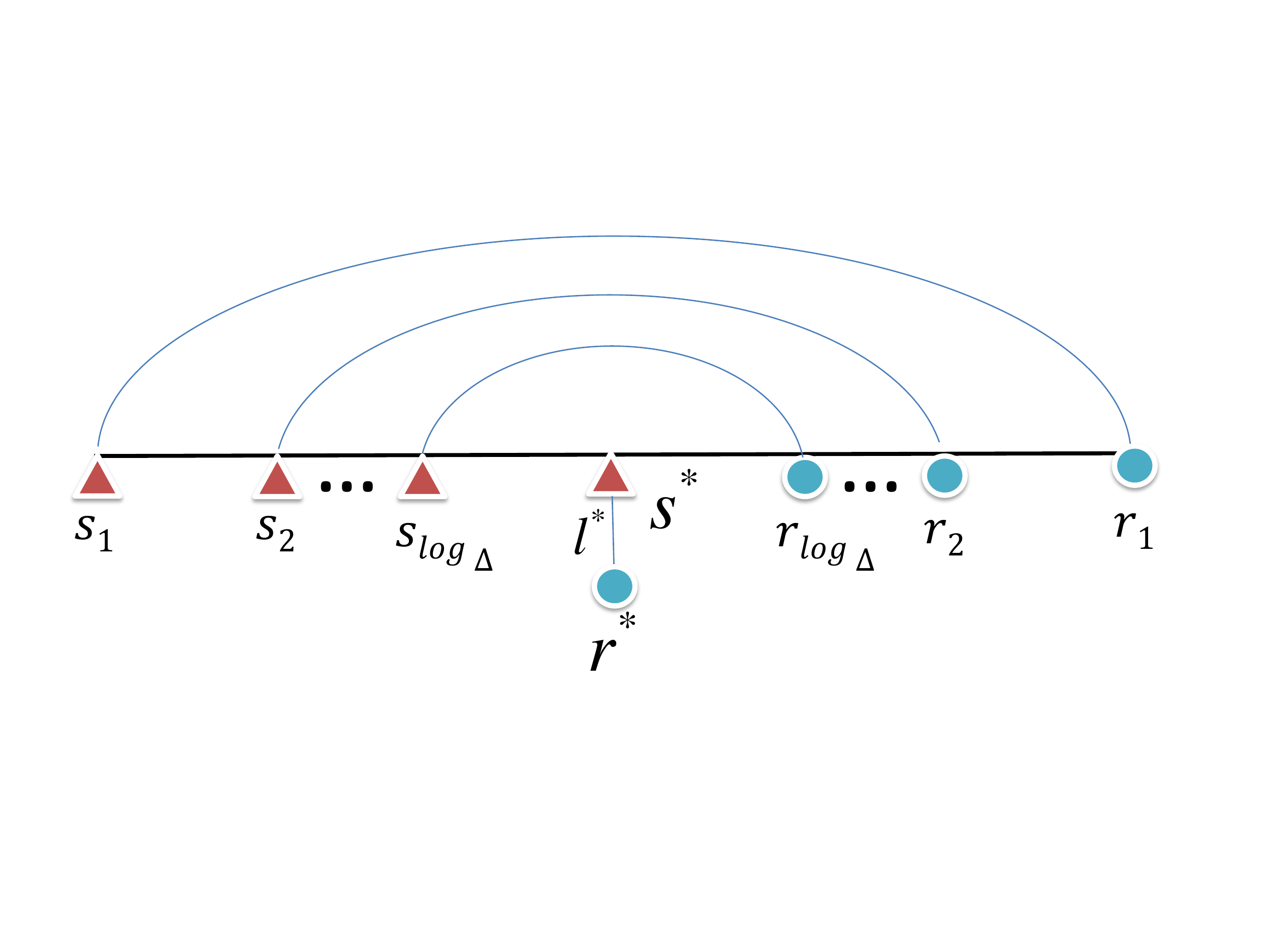}
\caption{ \label{fig:nash_pc}
\sf
Price of anarchy in the PC setting: a network in which the gap between the optimal solution and a no-regret solution is $\Omega(\log \Delta)$.
}
\end{center}
\end{figure}
}
\APPENDPCPA

\section{Power Control with Interference Cancellation (PIC)}
\label{sec:IC_PC}
In this section we consider games in the power control with IC setting where \emph{transmitters} can adopt their transmission energy in the range of $[1, P_{\max}]$ and in addition, \emph{receivers} can employ interference cancelation. This setting is denote as $PIC$ (power control+IC).
\def\APPENDPCUNEQ{
We begin by showing the following observation that should be contrasted with Obs. \ref{obs:nash_pc_contain} of the previous section.
\begin{observation}
\label{obs:ic_pc}
There exists a set of links $L$ for which
$\NoRegret^{IC}(L) \nsubseteq \NoRegret^{PIC}(L)$
\end{observation}
\begin{proof}
Let $\beta=2$, $\Noise=0$, $P_{\max} \geq 2$ and $L=\{\ell_1, \ell_2\}$ where $s_1,s_2$ are located at $(0,0)$ and $r_1,r_2$ are located at $(1,0)$.
We first show that there exists a $\epsilon$-regret history in the IC setting (with uniform powers) such that the average number of successful transmission is 1.
Consider an $2T$-round history in which $s_1$ transmits at the first $T$ rounds and $s_2$ transmits at the last $T$ rounds. Hence, the average utility of $s_1$ and $s_2$ is $1/2$. We now show that in every $\epsilon$-regret history the average utility of each of the transmitters should be at least
$-\epsilon$.  Let $p_1$ (resp. $p_2$) be the fraction of time that $s_1$ (resp., $s_2$) transmits in an $\epsilon$-regret history. If $s_1$ transmits in every round,  since $s_1$ and $s_2$ are equidistance from $r_1$, cancellation cannot help and hence it would succeed only in at most $1-p_2$ fraction of the time and its average utility is $1-2p_2$. Since $p_2=1/2$, it would have average utility of $0$, and in $\epsilon$-regret history, its average utility is at least $-\epsilon$. By symmetry, $s_2$ would also have average utility of at least $-\epsilon$. Concluding that the given history is indeed no-regret. We now claim that the given $(1/2,1/2)$ history is not a $\epsilon$-regret history for the PIC setting.
To see this, let $P_1 \in [1, P_{\max}]$ be the power assignment of $s_1$, then by letting $P_2=2P_1$ if $P_1\leq P_{\max}/2$ and $P_2=1$ otherwise, it holds that $\ell_2$ always succeeds when transmitting hence the average utility of $s_2$ is $1$ and in $\epsilon$-regret history it should be at least $1-\epsilon<1/2$ for $\epsilon<1/2$. Concluding that letting $s_2$ transmit for half of the rounds is \emph{invalid} no-regret strategy.
\QED \end{proof}
}
We show that Braess's paradox can once again happen and begin by comparing the PIC setting to the setting of power control without IC and to the most basic setting of uniform powers. 
\begin{lemma}
\label{lemma:PIC_pcunilower}
There exists a set of links $L$ and constant $c>1$ such that\\
(a) $\NoRegret^{PIC}_{\min}(L) \leq \NoRegret^{PC}_{\min}(L)/c$.\\
(b) $\NoRegret^{PIC}_{\min}(L) \leq \NoRegret_{\min}(L)/c$.\\
\end{lemma}
\def\APPENDPICPCUN{
\begin{proof}
Consider the $6$-transmitters network illustrated in Fig. \ref{fig:nash_pic}. Let $\alpha=8$, $\beta=10$, $\Noise=0$, $b=25^{1/8}$ and $P_{\max}=2$. Set $\widehat{\epsilon}<<1$ to be a sufficiently small constant. Intuitively, $\widehat{\epsilon}$ should be sufficiently small so that the receivers $r_4$, $r_5$ and $r_6$ consider the set of transmitters $s_1, s_2$ and $s_3$ are co-located at the same point, say at the point of transmitter $s_1$. For simplicity, we therefore analyze the network example as if this is the case while keeping in mind that this effect can be achieved by setting $\widehat{\epsilon}$ to be sufficiently small. In addition, if one insists on minimal transmitter-receiver distance $d(s_i,r_j)\geq 1$ , then the entire set of distances can be multiplied by $1/\widehat{\epsilon}$, without affecting the analysis (since there is no ambient noise, such normalization factor cancels out).
We begin by considering the value of solutions in an $\epsilon$-regret histories without IC. Let $p_i$ be the fraction of time that $s_i$ transmits in the worst PC setting for every $i \in \{1, \ldots, 6\}$.

It is easy to see that link $1$ will always be successful, since the length of the link $\widehat{\epsilon}$ is set to be sufficiently small.
Since $\ell_1$ has at most $\epsilon$-regret, this implies that $p_1 \geq 1- \epsilon$.

On the other hand, whenever $s_1$ transmits  it is clear that the link $\ell_2$ cannot be successful even if $s_2$ transmits with full power and $s_1$ transmits with power $1$ as $\beta/\widehat{\epsilon}^8> P_{\max}/(\beta^{16}\cdot \widehat{\epsilon}^8)$.

So if $s_2$ transmitted every time it would have average utility at most $-(1-\epsilon) + \epsilon = -1 + 2\epsilon < 0$ (since $\epsilon < 1/2$), while if it never transmitted it would have average utility $0$.  Thus its average utility is at least $-\epsilon$.  Since it can succeed only an $\epsilon$ fraction of the time (when link $1$ is not transmitting), we have that $\epsilon - (p_2 - \epsilon) \geq -\epsilon$ and thus $p_2 \leq 3\epsilon$. In the same manner, it also holds that $p_3 \leq 3 \epsilon$.

Now consider link $\ell_4$. As long as links $\ell_2$ and $\ell_3$ do not transmit $\ell_4$ always succeeds even if it transmits with power $P_4=1$ and $s_1, s_5,s_6$ transmit with full power $P_{\max}$. This holds since the amount of interference it suffers is at most $2/b^8+4/(16b^8)\leq 1/\beta$.
Since the fraction of time that both $\ell_2$ and $\ell_3$ are transmitting is at most $p_2+p_3 \leq 6\epsilon$, it holds that if $\ell_4$ always transmits it succeeds at least $1-6\epsilon$ fraction of the time and hence its average utility is $1-6\epsilon-6\epsilon=1-12\epsilon$ which is strictly positive by taking a sufficiently small $\epsilon$. Therefore it holds that in $\epsilon$-regret history, the average utility of $\ell_4$ is at least $1-13\epsilon$, concluding that  $p_4\geq 1-13\epsilon$.
By symmetry, the same holds for $\ell_5$ and $\ell_6$.
Overall, the value of any no-regret solution in PC setting is at least $\NoRegret^{PC}_{\min}(L)=\sum_{i=1}^6 p_i \geq 1-\epsilon+6\epsilon+3-39\epsilon \geq 4+o(1)$.

Let us now analyze what happens when using interference cancellation with power control and bound $\NoRegret^{PIC}_{\min}(L)$.  In this case, it is sufficient to consider a specific pure Nash solution. Let $N=\{s_1, s_2, s_3\}$ be transmitting with full power $P_{\max}$. It then holds that $r_2$ and $r_3$ can cancel the signal of $s_1$ and that both $r_2$ and $r_3$ can decode (resp., cancel) the signal of $s_2$ (this is achieved since $s_1,s_2$ and $s_3$ form an exponential chain with respect to $r_1,r_2,r_3$ though it looks the same with respect to the receivers of $r_4,r_5,r_6$) . We now show that in this case, $s_4$ cannot succeed even if they transmit with full power $P_{\max}$. This holds since $\beta \cdot 6/b^8 >P_{\max}$.
The same holds also for $s_5$ and $s_6$, concluding that $\NoRegret^{PIC}_{\min}(L) \leq 3$ and that $\NoRegret^{PC}_{\min}(L)/\NoRegret^{PIC}_{\min}(L)=4/3+o(1)$ as required. Consider claim (b).  This case is practically the same as the case of claim (a), in particular it holds that $\NoRegret^{PC}_{\min}(L)=\NoRegret_{\min}(L)$ (since $\ell_1,\ell_4, \ell_5,\ell_6$ always succeed when $s_2$ and $s_3$ are not transmitting and without IC, $s_2$ and $s_3$ cannot succeed if $s_1$ transmits). Therefore it also holds that $\NoRegret_{\min}(L)/\NoRegret^{PIC}_{\min}(L)=4/3+o(1)$. The lemma follows.
\QED \end{proof}

\begin{figure}[h!]
\begin{center}
\includegraphics[scale=0.3]{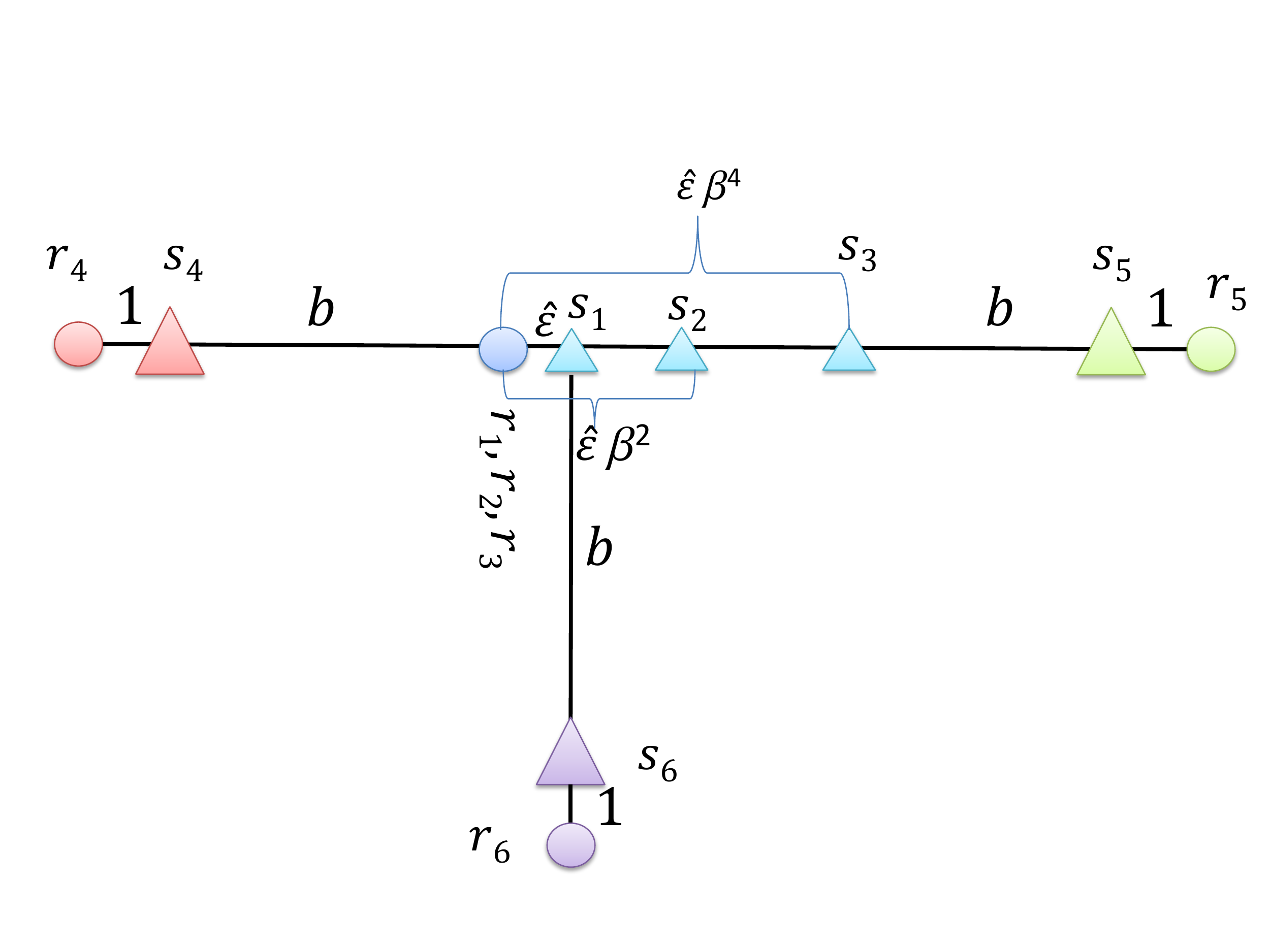}
\caption{ \label{fig:nash_pic}
\sf
Schematic illustration of a network in which playing IC with power control might generate no-regret solutions that are
worse by a factor of $\Omega(1)$ than no-regret solution in a setting without interference cancellation, with or without power control.
}
\end{center}
\end{figure}
}
\APPENDPICPCUN
Moreover, we proceed by showing that PIC can hurt the network by more than a constant when comparing PIC equilibria to IC equilibria. For an illustration of such a network, see Fig. \ref{fig:nash_all}(c).
\begin{theorem}
\label{theorem:IC_pc_lower}
There exists a set of links $L$ and constant $c>1$ such that the best pure Nash solution with PIC is worse by a factor of $\Omega(\log \Delta)$ than the worst no-regret solution with IC.
\end{theorem}
\def\APPENDPICIC{
\begin{proof}
Consider the following $m$-transmitters network for $m=\Omega(\log \Delta)$ described in Fig. \ref{fig:nash_all}(c). There is a set of $m-4$ receivers $R'=\{r_5, \ldots, r_m\}$ located at $(0,0)$.
Their corresponding transmitters $s_i$ are located at $(\sqrt{10}(1+2^{i/2}),0)$ for every $i \in \{5, \ldots, m\}$. The remaining first links are located as follows: $s_1$ is located at $(-3,1)$, $s_2$ is located at $(-3,-1)$, the receiver $r_2$ is located in the middle, between $s_1$ and $s_2$ and the receiver $r_1$ is located on top of $s_1$ at $(-3,2)$. The receiver $r_4$ is located at $(\sqrt{19},0)$, its transmitter $s_4$ is located at $(\sqrt{19},-1)$. Finally, the transmitter $s_3$ is located at $(\sqrt{19},1)$ and its receiver $r_3$ is located at $(\sqrt{19},2)$.
Let $\beta=1.5, \Noise=0, \alpha=2, P_{\max}=2$.

We begin by considering the PIC setting and showing that in this case there exists a unique Nash in which only $s_1,...,s_4$ are transmitting. Hence, the total value of
of any pure Nash is $4$.
First, observe that $s_1$ and $s_3$ always succeed even if they transmit with power $1$ and all other links transmit with power $P_{\max}$. For example, for $s_1$ we get that the received signal strength is $1$, the amount of interference from $s_2$ is at most $P_{2,1}=2/9$ and $P_{3,1}+P_{4,1}\leq 4/20$ in addition the amount of interference from all other transmitters $s_5,...,s_m$ is at most $1/4^4$. Hence $1 \geq \beta(2/9+2/20+1/4^4)$. Next, note that also $s_2$ and $s_4$ always succeed. Consider $s_2$ for example. There are two cases. If $s_1$ transmits with power $P_1=1$ then with interference cancellation, $s_2$ can succeed if it transmits with power $2$. This holds since in this case $P_{2,2}=2$, $P_{1,2}=1$, $\sum_{i \geq 3} P_{i}\leq 4/20+1/4^4$, hence
$r_2$ can successfully decode $s_2$ as $P_{2,2} \geq \beta(1+21/80)$. Alternatively, if $s_1$ transmits with $P_1=2$ and $P_2=1$ then $r_2$ can successfully cancel $s_1$ (by the same computation as above) and subsequently it can decode $s_2$ since in this case the received signal strength of $s_2$ is $1 \geq \beta(4/20+1/4^4)$. By the same argumentation we also have that $s_4$ can always succeed as well. Note that any Nash solution has the following structure: exactly one of the transmitters $s_1$ and $s_2$ transmit with power $1$ and the other with power $2$, and exactly one of the transmitters $s_3$ and $s_4$ transmit with power $1$ and the other with power $2$. We now show that $s_i$ always fails for every $i \geq 5$. Consider some pure Nash solution, then by the above, $s_1,...,s_4$ are active, without loss of generality, let $s_1$ be transmitting with power $1$ and $s_3$ be transmitting with power $2$. It then holds that the received
signal strength of transmitter $s_1$ at $r_i$ is $P_{1,i}=1/10$ and the received signal strength of $s_3$ is $P_{3,i}=2/20$. Hence, $P_{1,i}=P_{3,i}$ and in addition, these signals are stronger than any other signal, i.e., $P_{3,i}>P_{\max}/(10\cdot 2^{2j})$ for every $j \geq 5$. Hence, $r_i$ cannot cancel the strongest signal (as $P_{1,i} < \beta \cdot P_{3,i}$) and in particular it cannot decode its intended message. Concluding that any Nash solution in the IC setting consists of exactly $4$ links.
\par We proceed by considering the worst no-regret solution in the IC setting (without power control). Let $p_i$ be the fraction of time that $s_i$ transmits in an IC history. Note that since the links of $s_1$ and $s_3$ are sufficiently short they always succeed if they attempt to transmit hence in any $\epsilon$-regret history, they transmit at least $p_1, p_3 \geq 1-\epsilon$ fraction of the time. Consider link $\ell_2$ and note that it cannot succeed if $s_1$ transmits since $P_{1,2}=P_{2,2}$ hence if $s_2$ transmitted every time it would have average utility at most $-(1-\epsilon) + \epsilon = -1 + 2\epsilon < 0$ (since $\epsilon < 1/2$), while if it never transmitted it would have average utility $0$.  Thus its average utility is at least $-\epsilon$.  Since it can succeed only an $\epsilon$ fraction of the time (when link $1$ is not transmitting), we have that $\epsilon - (p_2 - \epsilon) \geq -\epsilon$ and thus $p_2 \leq 3\epsilon$. The same holds for $p_4$. Finally, consider some link $\ell_j$ for $j \geq 5$. Note that $\ell_j$ can always succeed if $s_2$ and $s_4$ are not transmitting. This holds since $P_{1,j} \geq \beta (P_{3,j}+1/32)$ (i.e., hence, $s_j$ can cancel the strongest signal of $s_1$), $P_{3,j} \geq \beta \cdot 1/32$ (i.e., it can successfully cancel the second strongest signal) and in addition by the structure of the exponential transmitters chain, the remaining of cancelations of $s_5,...,s_j$ are successful.
Hence, if $s_j$ transmitted every time it would have average utility at $1-6\epsilon-6\epsilon$ and therefore it transmits at least $p_5 \geq 1-13\epsilon$ fraction of the time. Since it holds for every $\ell_j$ for $j \geq 5$ we have that the total value in any $\epsilon$-regret history is at least $\NoRegret^{IC}_{\min}=\sum_{i\geq 5} p_i =\Omega(\log \Delta)-o(1)$. Since the total value of the unique pure Nash in the PIC setting is 4, the lemma follows.
\QED \end{proof}
}
\APPENDPICIC

\begin{corollary}
There exists a set of links $L$ satisfying that $\NoRegret^{PIC}_{\min}(L) \leq (c/\log \Delta)\cdot \NoRegret^{IC}_{\min}(L)$.
\end{corollary}
As in the previous sections, we show that our examples are essentially tight.
\begin{theorem}
\label{lem:nash_ic_pc_upper}
For every set of links $L$ it holds that there exists a constant $c\geq 1$ such that \\
(a) $\NoRegret_{\min}^{PIC}(L) \geq \NoRegret_{\max}(L)/c$.\\
(b) $\NoRegret_{\min}^{PIC}(L) \geq \NoRegret_{\max}^{PC}(L)/c$.\\
(c) $\NoRegret_{\min}^{PIC}(L)  \geq \NoRegret_{\max}^{IC}(L)/(c \log \Delta)$.
\end{theorem}
\def\APPANDPICUPPER{
\begin{proof}
Part (a) follows by the argumentation as in Lemma \ref{lem:ic_no_ic_feas}
and Part (b) follows trivially by Part (a) and by Theorem~\ref{thm:nash_pc_upper}.
\par Consider part (c).
Fix an arbitrary $\epsilon$-regret history with IC and uniform powers and let $p_i$ be the fraction of rounds in which $s_i$ attempts to transmit.  Similarly, fix an arbitrary $\epsilon$-regret history with PIC, and let $q_i$ be the fraction of rounds in which $s_i$ attempts to transmit (at any nonzero power).  By Lemma~\ref{lem:attempts}, it is sufficient to prove that $\sum_{i \in L} p_i \leq O(\sum_{i \in L} q_i)$.
Note that since the average number of successful connections in the best history of the IC case is $\sum_{i \in L} p_i$, there must exist some set of connections $A \subseteq L$ that transmitted successfully in some round $t \in [1, T]$ such that $|A| \geq \sum_{i \in L} p_i$ and $A$ is feasible when all transmitters transmit with power $1$ and employing IC. Then by Cl. \ref{lem:ic_no_ic_feas}, there exists a subset $\widetilde{A} \subseteq A$ of cardinality $|\widetilde{A}|\geq |A|/O(\log \Delta)$ that is feasible without IC and with uniform power level $P_{\max}$.
Let $B = \{i : q_i \geq 1/2\}$ and let $A' = \widetilde{A} \setminus B$.  If $|B \cap \widetilde{A}| \geq |\widetilde{A}| / 2$ then we are done, since then
$$O(\log \Delta) \cdot \sum_{i \in L} p_i \leq |\widetilde{A}| \leq 2 |B| \leq 2\sum_{i} q_i$$ as required.
So without loss of generality we will assume that $|B \cap \widetilde{A}| < |\widetilde{A}|/2$, and thus that $|A'| > |\widetilde{A}|/2$.  Note that $A'$ is a subset of $\widetilde{A}$, and so it is feasible in the IC setting.
Now let $\widehat{A} = \{i : \sum_{j \in A'} a_i(j) \leq 2\}$ be an amenable subset of $A'$.  By Fact \ref{fc:amenable}(a), it holds that $\widehat{A} \geq |A'|/ 2 \geq |A| / 4$.  By Fact \ref{fc:amenable}(b) then implies that for any link $\ell_i \in L$, its total affectance on $\widehat{A}$ is small: $a_{i}(\widehat{A})=\sum_{j \in \widehat{A}} a_i(j) \leq c'$ for some constant $c\geq 0$.
Note that Fact \ref{fc:amenable}(b) holds for the case where all interfering transmitters transmit with the same power $P_{\max}$. Since in the power control setting $P_j \leq P_{\max}$, the actual affectance can only be lower.
We have that
\begin{equation} \label{eq:PIC1}
\sum_{i \in L} \sum_{j \in \widehat{A}} q_i a_i(j) \leq c' \cdot \left(\sum_{i \in L} q_i\right),
\end{equation}
where $\sum_{j \in \widehat{A}} q_i a_i(j)$ is the average affectance in the PIC history if all links that attempt to transmit, transmit with full power. The inequality follow by Fact \ref{fc:amenable}(b).
On the other hand, we know that the $q_i$ values correspond to the worst $\epsilon$-regret history in the PIC setting.
Consider some $j \in \widehat{A}$, hence $q_j < 1/2$ (i.e., the average utility of link $\ell_j$ is at most $1/2$). Let $y_j$ be the fraction of time $s_j$ would have succeeded has it transmitted in every round with full power $P_{\max}$ and in the IC setting. Since the average utility of the best single action is at most $1/2+\epsilon$ in particular the average utility of transmitting with full power $P_{\max}$ is at most $1/2+\epsilon$ as well. It therefore holds that  $y_j-(1-y_j)\leq 1/2+\epsilon$ or the that $y_j \leq \frac34 + \frac{\epsilon}{2}$.
In other words, in at least $1-y_j=\frac14 - \frac{\epsilon}{2}$ fraction of the rounds the affectance of the other links on the link $\ell_j$ must be at least $1$ (or else $j$ could succeed in those rounds).
Thus the average affectance (over the length of the history) on $\ell_j$ (assuming it transmits with full power $P_{\max}$) is at least $\frac14 - \frac{\epsilon}{2}$.  Summing over all $j \in \widehat{A}$, to get that
\begin{equation} \label{eq:PIC2}
\sum_{j \in \widehat{A}} \sum_{i \in L} q_i a_i(j) \geq \frac{1-2\epsilon}{4} \geq \Omega(|\widehat{A}|).
\end{equation}
Combining equations~(\ref{eq:PIC1}) and~(\ref{eq:PIC2}) (and switching the order of summations) implies that $|\widehat{A}| \leq O(\sum_{i \in L} q_i)$.  Since $|\widehat{A}| \geq |A|/O(\log \Delta) \geq \sum_{i \in L} p_i/O(\log \Delta)$ as desired.
\QED \end{proof}
}
\APPANDPICUPPER
Finally, as a direct consequences of our result, we obtain a tight bound for the price of total anarchy in the PIC setting.
\begin{corollary}
\label{cor:pic_price}
For every set of links $L$  it holds that the price of total anarchy with PIC is $\Theta(\log (\Delta \cdot P_{\max}))$.
\end{corollary}
\def\APENNEDPICCOR{
\begin{proof}
Let $OPT \subseteq L$ denote some optimal solution without PIC, i.e., the set of transmitters forming a maximum $\beta$-feasible set, and let $OPT^{PIC} \subseteq L$ denote some optimal solution with PIC.
We first show that
$|OPT^{PIC}|/\mathcal{N}^{PIC}_{\min}(L))=O(\log (\Delta \cdot P_{\max}))$. According to Theorem 2 of \cite{Asgeirsson11}, it holds that
$\mathcal{N}_{\min}(L)\geq |OPT|/c'$ for some constant $c'>1$.
Hence, by combining with Theorem~\ref{lem:nash_ic_pc_upper}(a) we get that
\begin{equation*}
\mathcal{N}^{PIC}_{\min}(L) \geq \mathcal{N}_{\min}(L)/c \geq |OPT|/(c \cdot c') \geq |OPT^{PIC}|/O(\log (\Delta \cdot P_{\max})),
\end{equation*}
where the last inequality follows by Lemma \ref{lem:ic_no_ic_feas}.
Finally, we show that this is tight by noting that the example of Fig. \ref{fig:poa_ic} can be easily modified so that there are $m=\Omega(\log \Delta \cdot P_{\max})$ transmitters (the location and powers should be carefully set) whose receivers are positioned at the origin. By placing the additional
transmitters $\widetilde{s}_1$ and $\widetilde{s}_2$
at equidistance from the origin, we get that these two links block the cancellation sequence at the $m$ receivers; in addition
by setting these links to be sufficiently short, we get that in $\widetilde{s}_1$ and $\widetilde{s}_2$ transmit most of the time. Thus, any $\epsilon$-regret history is of total value of $2+o(1/n)$ but $OPT^{PIC}=\Omega(\log (\Delta \cdot P_{\max}))$.
\QED \end{proof}
}
\APENNEDPICCOR
\section{Decreasing the SINR Threshold}
%
%
We begin by showing that in certain cases the ability to successfully decode a message at a lower SINR threshold results in \emph{every} no-regret solution having lower value than \emph{any} no-regret solution at higher $\beta$. For an illustration of such a network, see Fig. \ref{fig:nash_all}(d).

\begin{theorem} \label{thm:beta_nash_lb}
There exists a set of links $L$ and constants $1 < \beta' < \beta$ such that $\NoRegret^{\beta'}_{\max}(L) \leq \NoRegret^{\beta}_{\min}(L)/c$ for some constant $c>1$.
\end{theorem}
\def\APPENDBETALOWER{
\begin{proof}
Let $\beta'$ be slightly greater than $1$ (say $101/100$), and let $\beta = 4$.  Let $L$ be as in Fig. \ref{fig:nash_all}(d) with $\alpha = 2$, some small constant noise $\Noise$, and $c = \sqrt{1/(\beta \Noise)}$ and $a=b=\frac23 c$.

We first analyze $\NoRegret^{\beta}_{\min}(L)$.  Fix an $\epsilon$-regret history under threshold $\beta$. Note that the SINR of $\ell_3$ is at most $\beta$, simply due to background noise.  So by making $\beta$ infinitesimally larger, $\ell_3$ can never succeed, and thus it transmits in at most an $\epsilon$ fraction of the times.  Whenever $\ell_3$ does not transmit, $\ell_1$ can succeed no matter what $\ell_2$ does, since the SINR of $\ell_1$ is at least
\begin{equation*}
\frac{1/a^2}{\Noise + \frac{1}{(3a)^2}} = \frac{(9/4)\beta \Noise}{\Noise + (1/4)\beta\Noise} = \frac{9}{2} > \beta.
\end{equation*}
Hence choosing to transmit in every round would get $\ell_1$ utility at least $1-2\epsilon$, and thus it must succeed in at least $1-3\epsilon$ fraction of the times.  The same is true for $\ell_2$ by symmetry.  So we have that in any $\epsilon$-regret history under threshold $\beta$, the average number of successes is at least $2-6\epsilon$.

We now analyze $\NoRegret^{\beta'}_{\max}(L)$.  Fix an $\epsilon$-regret history under threshold $\beta'$.  We first claim that $\ell_3$ can always succeed, no matter what $\ell_1$ and $\ell_2$ do.  This is because its SINR is at least
\begin{equation*}
\frac{1/c^2}{\Noise + \frac{2}{(a+b)^2 + c^2}} = \frac{1/c^2}{\Noise + \frac{2}{25c^2 / 9}} = \frac{1/c^2}{\Noise + \frac{18}{25} \frac{1}{c^2}} = \frac{\beta \Noise}{\Noise + \frac{18}{25} \beta \Noise} = \frac{4}{1 + (72/25)} = \frac{100}{97} > \beta'.
\end{equation*}
Thus $\ell_3$ must transmit in at least a $1-\epsilon$ fraction of the rounds.  In any round where $\ell_3$ transmits, neither $\ell_1$ nor $\ell_2$ can succeed since $a=b$ and $\beta' \geq 1$.  Thus the average number of successes is at most $(1-\epsilon) + 2\epsilon = 1+\epsilon$.
\QED \end{proof}

}
\APPENDBETALOWER

\def\APPANDNOISEEFFECT{
For simplicity, for the rest of the proof we will assume that $P_{vv}\geq 2\beta \cdot \Noise$ for every value of $\beta$ we consider.  This does not hold in the bad example of Theorem~\ref{thm:beta_nash_lb}, but that example can easily be modified so that it holds (although it makes the example messier).  This is a standard assumption, and limits strange phenomena due to thresholding.
In addition, we isolate the affect of the SINR threshold value and  restrict attention to uniform powers.

Given an SINR threshold $\beta'$ and links $\ell_w$ and $\ell_v$ define $$a_{w}^{\beta'}(v)=\min\left\{1, c_v(\beta') P_{w,v}/P_{v,v} \right\}.$$ The affectness of a subset of (resp. on) a subset of links $L'$ is denoted by $a^{\beta'}_{L'}(w)$ ($a^{\beta'}_{w}(L')$).
We now show a simple lemma which will help us prove that the above example is tight.
\begin{observation}
\label{obs:wo_noise}
Let $L$ be a $2\beta$-feasible set without noise, $\Noise=0$ then $L$ is a $\beta$-feasible set with noise $\Noise>0$ satisfying that $P_{vv} \geq 2\beta \Noise$ for every $\ell_v \in L$.
\end{observation}
Since $L$ is a $2\beta$-feasible set, it holds that $\sum_{u \in L}P_{uv}/P_{vv} \leq 1/(2\beta)$ for every $\ell_v \in L$. Since $P_{vv} \geq 2\beta \Noise$, it holds that $1/(2\beta) \leq 1/\beta-\Noise/P_{vv}$ and hence that
$\sum_{u \in L}P_{uv}/P_{vv} \leq 1/\beta-\Noise/P_{vv}$.
Concluding that $L$ is a $\beta$-feasible set with $\Noise>0$.
}

We now show that the gap between the values of no-regret solution for different SINR threshold values is bounded by a constant.
\begin{lemma}
\label{lem:beta_nash_lb}
For every $1 \leq \beta' \leq \beta$ and every set of links $L$ satisfying that $P_{vv}\geq 2\beta \cdot \Noise$ for every $\ell_v \in L$, it holds that $\NoRegret_{\min}^{\beta'}(L) \geq \NoRegret_{\max}^{\beta}(L)/c$ for some constant $c\geq 1$.
\end{lemma}
\def\APPENDBETANASH{
\begin{proof}
\APPANDNOISEEFFECT

Consider any two no-regret solutions, one for the $\beta$ setting and the other for the $\beta'$ setting.
Let $p_i$ (resp., $q_i$) be the fraction of times in which $s_i$ attempts to transmit in a $\epsilon$-regret history with SINR threshold $\beta$ (resp., $\beta'$).
By taking $\epsilon =O(1/n)$ and using Lemma~\ref{lem:attempts}, it suffices to show that $\sum_i p_i \leq c \cdot \sum_i q_i$.
Note that since the average number of successful connections in the $\beta$ setting is $\sum p_i$, there must exist some set of connections $A \subseteq L$ that transmitted successfully in some round $t \in [1, T]$ such that $|A| \geq \sum p_i$ and $A$ is feasible with SINR threshold $\beta$.
Let $B = \{i : q_i \geq 1/2\}$ and let $A' = A \setminus B$.  If $|B \cap A| \geq |A| / 2$ then we are done, since then $$\sum_i p_i \leq |A| \leq 2|B| \leq 2 \sum_{i} q_i$$ as required.
So without loss of generality we will assume that $|B \cap A| < |A|/2$, and thus that $|A'| > |A|/2$.  Note that $A'$ is a subset of $A$, and so it is $\beta$-feasible as well.
Since it is $\beta$-feasible for some $\Noise\geq 0$, it is $\beta$-feasible also for $\Noise=0$.
Now let $\widehat{A} = \{i : \sum_{j \in A'} a_i(j) \leq 2\}$ be an emanable subset of $A'$.  By Fact \ref{fc:amenable}(a), we have that $\widehat{A} \geq |A'|/ 2 \geq |A| / 4$.  By Fact \ref{fc:amenable}(b) then implies that for any link $i \in L$, its total affectance on $A$ is small: $\sum_{j \in A} a^{\beta}_i(j) \leq c'$ for some constant $c\geq 0$.
For a subset $L' \subseteq L$ and some $\widehat{\beta}>0$ define
$$\widehat{E}^{\widehat{\beta}}(L')=\widehat{\beta} \cdot \sum_{i \in L}\sum_{j  \in L'} q_i \cdot P_{j,i}/P_{j,j}.$$
Thus we have that
\begin{eqnarray}
\label{eq:beta1}
\widehat{E}^{\beta}(\widehat{A}) &\leq&
\sum_{i \in L}\sum_{j  \in \widehat{A}}
a^{\beta}_{i}(j)\leq c \cdot \left(\sum_{i \in L} q_i\right).
\end{eqnarray}
On the other hand, we know that the $q_i$ values correspond to the worst history with SINR threshold $\beta'$ in which every link has regret at most $\epsilon$.  Let $j \in A'$.  Then $q_j < 1/2$, which means the average utility of link $\ell_j$ is at most $1/2$. Let $y_j$ be the fraction of time $s_j$ would have succeeded (i.e., the SINR value at its receiver $r_j$ is above $\beta'$) has it transmitted in every round. Since the average utility of the best single action is at most $1/2+\epsilon$ it holds that $y_j-(1-y_j)\leq 1/2+\epsilon$ or the that $y_j \leq \frac34 + \frac{\epsilon}{2}$.
In other words, in at least $1-y_j=\frac14 - \frac{\epsilon}{2}$ fraction of the rounds the affectance $a^{\beta'}$ of the other links on the link $\ell_j$ must be at least $1$ (or else $j$ could succeed in those rounds).  Thus the expected affectance (taken over a random choice of time slot) on $\ell_j$ is at least $\sum_{i \in L} a^{\beta'}_i(j) q_i \geq \frac14 - \frac{\epsilon}{2}$.  Summing over all $j \in \widehat{A}$, to get that
\begin{eqnarray} \label{eq:beta2}
\widehat{E}^{2\beta'}(\widehat{A})&\geq&\sum_{j \in \widehat{A}} \sum_{i \in L} a^{\beta'}_i(j) q_i \geq \sum_{j \in A} \frac{1-2\epsilon}{4} \geq c' \cdot |\widehat{A}|,
\end{eqnarray}
where the first inequality follows by Obs. \ref{obs:wo_noise}.
Combining equations~(\ref{eq:beta1}) and~(\ref{eq:beta2}) (and switching the order of summations) implies that
\begin{eqnarray*}
c' \cdot |\widehat{A}| &\leq& \widehat{E}^{2\beta'}(\widehat{A}) = 2\beta'/\beta \cdot \widehat{E}^{\beta}(\widehat{A})\leq
c \cdot 2\beta'/\beta\cdot N^{\beta'}_{\min}(L)~.
\end{eqnarray*}
Since $|\widehat{A}| \geq |A|/4 \geq \Omega(\NoRegret^{\beta}_{\max}(L))$, we get that $O(\beta/\beta') \cdot \NoRegret^{\beta}_{\max}(L) \leq
\NoRegret^{\beta'}_{\min}(L)$ as desired.
The lemma follows.

\QED \end{proof} 
}
\APPANDNOISEEFFECT
%

\section{Conclusion}
In this paper we have shown that Braess's paradox can strike in wireless networks in the SINR model: improving technology can result in worse performance, where we measured performance by the average number of successful connections.  We considered adding power control, interference cancellation, both power control and interference cancellation, and decreasing the SINR threshold, and in all of them showed that game-theoretic equilibria can get worse with improved technology.  However, in all cases we bounded the damage that could be done.

There are several remaining interesting open problems.  First, what other examples of wireless technology exhibit the paradox?  Second, even just considering the technologies in this paper, it would be interesting to get a better understanding of when exactly the paradox occurs.  Can we characterize the network topologies that are susceptible?  Is it most topologies, or is it rare?  What about random wireless networks?  Finally, while our results are tight up to constants, it would be interesting to actually find tight constants so we know precisely how bad the paradox can be.

\def\thepage{}
{
\bibliographystyle{plain}
\bibliography{GT}
}

\end{document}